\documentclass{article}

\usepackage{listings}
\usepackage{graphicx}
\usepackage{tikz}
\usetikzlibrary{arrows, automata, patterns}
\usepackage{lipsum}
\usepackage{enumerate}
\usepackage{makecell}
\usepackage{amssymb}

\usepackage[unicode=true]{hyperref}

\usepackage[linesnumbered,ruled,vlined,boxed,algo2e]{algorithm2e}
\usepackage{nicefrac}
\usepackage{amsmath}
\usepackage{float}
\usepackage{mathtools}
\usepackage{amsthm}
\usepackage{multirow}
\usepackage{fullpage}
\usepackage{amsfonts}
\usepackage{clipboard}

\newtheorem{theorem}{Theorem}

\newtheorem{lemma}{Lemma}

\newtheorem{claim}{Claim}[lemma]

\newtheorem{subclaim}{Claim}[claim]

\newtheorem{subsubclaim}{Claim}[subclaim]

\newtheorem{corollary}{Corollary}[theorem]

\newtheorem{definition}{Definition}[theorem]

\newtheorem{problem}{Problem}

\newcommand{\Reminder}[1]{

\vspace{0.5em}
\noindent\textbf{Reminder of~\autoref{#1}.} \textit{\Paste{#1}}
\vspace{0.5em}

}

\usepackage{xcolor}
\definecolor{DarkGreen}{RGB}{1,50,32}
\usepackage{silence}
\usepackage{chngcntr}
\counterwithin*{equation}{section}
\counterwithin*{equation}{subsection}

\begin{document}
\ActivateWarningFilters[pdftoc]
\newcommand{\defeq}{:=}
\newcommand{\eps}{\varepsilon}
\newcommand{\ReturnCode}{\textbf{return}}

\newcommand{\sidford}[1]{{\color{ForestGreen} \textbf{Aaron}: #1}} 
\newcommand{\liam}[1]{{\color{red} \textbf{Liam}: #1}} 
\newcommand{\avi}[1]{{\color{purple} \textbf{Avi}: #1}} 
\newcommand{\vvw}[1]{{\color{Plum} \textbf{Virgi}: #1}} 
\newcommand{\uri}[1]{{\color{red} \textbf{Uri}: #1}} 

\newcommand{\blue}[1]{{\color{blue}#1}}

\newcommand{\codestyle}[1]{\texttt{#1}}
\newcommand{\Initialize}{\mbox{\codestyle{Initialize}}}
\newcommand{\Cycle}{\mbox{\codestyle{Cycle}}}
\newcommand{\RT}{\mbox{\codestyle{RT}}}
\renewcommand{\L}{\mbox{\codestyle{L}}}
\newcommand{\CycleOdd}{\codestyle{CycleOdd}}
\newcommand{\BallOrCycle}{\codestyle{BallOrCycle}}
\newcommand{\ClusterOrCycleBounded}{\codestyle{ClusterOrCycleBounded}}
\newcommand{\ClusterOrCycle}{\codestyle{ClusterOrCycle}}
\newcommand{\SimpleCycle}{\codestyle{SimpleCycle}}
\newcommand{\Next}{\codestyle{Next}}
\newcommand{\Sample}{\codestyle{Sample}}
\newcommand{\Dijkstra}{\codestyle{Dijkstra}}
\newcommand{\Preprocess}{\codestyle{Preprocess}}
\newcommand{\HashTable}{\codestyle{HashTable}}
\newcommand{\Heap}{\codestyle{Heap}}
\newcommand{\RelaxNext}{\codestyle{RelaxNext}}

\newcommand{\PreprocessGraph}{\codestyle{Initialize}}
\newcommand{\Route}{\codestyle{Route}}
\newcommand{\TreeRoute}{\codestyle{TreeRoute}}
\newcommand{\N}{\mathbb{N}}
\newcommand{\MinCycle}{\codestyle{MinCycle}}
\newcommand{\Query}{\codestyle{Query}}
\newcommand{\Ball}{\codestyle{Ball}}
\newcommand{\DistanceOracle}{\codestyle{TZ-DistanceOracle}}
\newcommand{\SparseOrCycle}{\codestyle{SparseOrCycle}}
\newcommand{\Intersection}{\codestyle{Intersection}}
\newcommand{\CycleAdditive}{\codestyle{CycleAdditive}}
\newcommand{\GenerateSi}{\codestyle{ComputeS}}

\newcommand{\codeNull}{\codestyle{null}}
\newcommand{\codeYes}{\codestyle{Yes}}
\newcommand{\codeNo}{\codestyle{No}}
\newcommand{\codeAnd}{~ \mathrm{and} ~}
\newcommand{\codeOr}{~ \mathrm{or} ~}
\newcommand{\wt}{\ell}
\newcommand{\Cl}{CL}
\newcommand{\CL}{CL}
\newcommand{\cl}{c\ell}

\newcommand{\EQ}{\;=\;}
\newcommand{\GE}{\;\ge\;}
\newcommand{\Ot}{\tilde{O}}
\newcommand{\stactri}{\stackrel\triangle}

\newcommand{\EE}{\mathbb{E}}
\newcommand{\RR}{\mathbb{R}}

\DeclarePairedDelimiter{\ceil}{\lceil}{\rceil}
\DeclarePairedDelimiter{\floor}{\lfloor}{\rfloor}
\DeclarePairedDelimiter{\pair}{\langle}{\rangle}

\author{Avi Kadria\thanks{Department of Computer Science, Bar Ilan University, Ramat Gan 5290002, Israel. E-mail {\tt avi.kadria3@gmail.com}.} \and Liam Roditty\thanks{Department of Computer Science, Bar Ilan University, Ramat Gan 5290002, Israel. E-mail {\tt liam.roditty@biu.ac.il}. Supported in part by BSF grants 2016365 and 2020356.}}

\title{Compact routing schemes in undirected and directed graphs}
\pagestyle{empty}

\maketitle

\begin{abstract}

In this paper, we study the problem of compact routing schemes in weighted undirected and directed graphs. 

\textit{For weighted undirected graphs}, more than a decade ago, Chechik [PODC'13] presented a $\approx3.68k$-stretch compact routing scheme that uses $\Ot(n^{1/k}\log{D})$ local storage, where $D$ is the normalized diameter, for every $k>1$. 
We present a $\approx 2.64k$-stretch compact routing scheme that uses $\Ot(n^{1/k})$ local storage \textit{on average} in each vertex. 
This is the first compact routing scheme that uses total local storage of $\Ot(n^{1+1/k})$ while achieving a $c \cdot k$ stretch, for a constant $c < 3$.

In real-world network protocols, messages are usually transmitted as part of a communication session between two parties. Therefore, 
more than two decades ago, Thorup and Zwick [SPAA'01] 
considered compact routing schemes that establish a communication session using a handshake. 
In their handshake-based compact routing scheme, the handshake is routed along a $(4k-5)$-stretch path, and the rest of the communication session is routed along an optimal $(2k-1)$-stretch path. It is straightforward to improve the $(4k-5)$-stretch of the handshake to $\approx3.68k$-stretch using the compact routing scheme of Chechik [PODC'13]. 
We improve the handshake stretch to the optimal $(2k-1)$, by borrowing the concept of roundtrip routing from directed graphs to \textit{undirected} graphs.

\textit{For weighted directed graphs}, more than two decades ago, Roditty, Thorup, and Zwick [SODA'02 and TALG'08] presented a $(4k+\eps)$-stretch compact roundtrip routing scheme that uses $\Ot(n^{1/k})$ local storage for every $k\ge 3$. For $k=3$, this gives a $(12+\eps)$-roundtrip stretch using $
\Ot(n^{1/3})$ local storage. We improve the stretch by developing a $7$-roundtrip stretch routing scheme with $\Ot(n^{1/3})$ local storage.
In addition, we consider graphs with bounded hop diameter and present an optimal $(2k-1)$-roundtrip stretch routing scheme that uses $\Ot(D_{HOP}\cdot n^{1/k})$, where $D_{HOP}$ is the hop diameter of the graph.
\end{abstract}
\thispagestyle{empty}
\clearpage

\pagenumbering{arabic}
\setcounter{page}{1}
\section{Introduction}
Routing is a fundamental task in computer networks. A \emph{routing scheme} is a mechanism designed to deliver messages efficiently from a source vertex to a destination vertex within the network.
In this paper, we study both undirected and directed weighted graphs, aiming to route along short paths.

More specifically, a routing scheme is composed of a preprocessing phase and a routing phase.
In the preprocessing phase, the entire graph is accessible, allowing the preprocessing algorithm to compute a routing table and a label for each vertex\footnote{In this paper, we study labeled routing schemes, where the preprocessing algorithm can assign labels to the vertices. When the vertex labels are fixed and cannot be changed, the routing scheme is called an name-independent routing scheme (see, for example,~\cite{DBLP:conf/spaa/AbrahamGMNT04, DBLP:conf/podc/KonjevodRX06, DBLP:journals/ipl/Laing07, DBLP:journals/talg/AbrahamGMNT08, DBLP:conf/wdag/GavoilleGHI13})}, which is then stored in the local storage of each vertex. 
In the routing phase, the routing algorithm at each vertex on the routing path can only access the local storage of the vertex. 
The routing algorithm  
gets as input a message, a destination label, and possibly a header, and decides which of the vertex neighbors is the next vertex on the routing path.
The routing continues until the message reaches the destination.

A \textit{compact} routing scheme is a routing scheme that uses $o(n)$ space in the local storage on average at each vertex, where 
$n$ is the number of vertices in the graph. 
Let $u$ be a source vertex and let $v$ be a destination vertex. 
We denote by $\hat{d}(u,v)$ the length of the path used by the routing algorithm to route a message from $u$ to $v$.
The \textit{stretch} of the routing scheme is defined as $\max_{u,v\in V}(\frac{\hat{d}(u,v)}{d(u,v)})$, where $d(u,v)$ is the distance from $u$ to $v$. The \textit{roundtrip} stretch is defined as $\max_{u,v\in V}(\frac{\hat{d}(u,v)+\hat{d}(v,u)}{d(u,v)+d(v,u)})$.

The design of efficient compact routing schemes in undirected graphs has been a well-studied subject in the last few decades, see for example~\cite{DBLP:journals/jacm/PelegU89,DBLP:journals/jal/AwerbuchBLP90,DBLP:journals/siamdm/AwerbuchP92,DBLP:journals/jal/Cowen01,DBLP:journals/jal/EilamGP03,DBLP:conf/spaa/ThorupZ01, DBLP:conf/podc/RodittyT15, DBLP:conf/podc/Chechik13}. In~\autoref{tab:compact-undirected} we summarize the previous results.

From the Erd\H{o}s girth conjecture, it follows that every routing scheme with stretch $<2k+1$ must use a total storage of $\Omega(n^{1+1/k})$ bits. 
The approximate distance oracle data structure of Thorup and Zwick~\cite{DBLP:journals/jacm/ThorupZ05}, which is implemented in the centralized model, where all the information is available upon a distance query to the data structure, has an optimal $(2k-1)$-stretch with $\Ot(n^{1+1/k})$ total storage.
In light of the gap between routing schemes and approximate distance oracles, the following problem is natural.
\begin{problem}\label{Pr-0}
    For every $k\ge 2$, given a weighted undirected graph, what is the best stretch of a routing scheme that uses $\Ot(n^{1+1/k})$ total storage?
\end{problem}
The $\approx3.68k$-stretch compact routing scheme of Chechik~\cite{DBLP:conf/podc/Chechik13}, from more than a decade ago, is the current best stretch with $\Ot(n^{1/k})$ worst-case local storage.
In this paper, we improve the stretch to $\approx 2.64k$ by allowing an average local storage of $\Ot(n^{1/k})$, as presented in the following theorem.
\begin{theorem} \label{T-Compact-Amortized-3k} \Copy{T-Compact-Amortized-3k} {
    Let $G=(V, E, w)$ be a weighted undirected graph. For every $k\ge 3$, there is an $\approx 2.64k$-stretch compact routing scheme that uses local routing tables of an average size of $\tilde{O}(n^{1/k})$, vertex labels of size $\tilde{O}(k)$ and packet headers of size $\tilde{O}(k)$.
}
\end{theorem}

All the compact routing schemes mentioned so far solve the problem of sending a single message from the source to the destination, while in most real-world network applications, two parties communicate over the network for a session. 
A \textit{communication session} is composed of two phases. In the first phase, a connection is established between the source and the destination, and in the second phase, a stream of messages is exchanged between the parties.
Many real-world protocols, such as  \texttt{TLS}, \texttt{QUIC}, \texttt{TCP}, \texttt{SSH}, \texttt{Wi-Fi}, and \texttt{Bluetooth}, adhere  to this framework. 

We consider the \textit{handshake} mechanism for the establishment phase, as presented by Thorup and Zwick~\cite{DBLP:conf/spaa/ThorupZ01}. The handshake is composed of two messages of size $\Ot(1)$, that are exchanged between the parties to establish the connection. 
The first message is sent from the source to the destination, and the second message is sent from the destination back to the source. Since the handshake is composed of a message sent from the source to the destination and back, the stretch of the handshake is the roundtrip stretch defined earlier.

Thorup and Zwick~\cite{DBLP:conf/spaa/ThorupZ01} presented a compact routing scheme that uses a handshake,
in which two messages are routed along a $(4k-5)$-stretch path,
to establish a connection. Then, a stream of messages is routed along an optimal $(2k-1)$-stretch path.
While the compact routing scheme of Chechik~\cite{DBLP:conf/podc/Chechik13} achieves an $\approx3.68k$-roundtrip stretch for the handshake, there is still 
a significant gap from the optimal $(2k-1)$-stretch followed by the Erd\H{o}s girth conjecture. 
Therefore, the main open problem for routing in a communication session is to reduce the stretch of the handshake phase and obtain a compact roundtrip routing scheme with improved stretch.
\begin{problem}\label{Pr-1}
    For every $k\ge 2$, given a weighted undirected graph, what is the best roundtrip stretch of a routing scheme that uses $\Ot(n^{1/k})$ local storage?
\end{problem}

In this paper we \textit{solve}~\autoref{Pr-1} by presenting an \textit{optimal} 
$(2k-1)$-roundtrip stretch for the handshake phase, that matches the lower bound that follows from the Erd\H{o}s girth conjecture, as presented in the following theorem.
\begin{theorem} \label{T-Compact-Undirected} \Copy{T-Compact-Undirected} {
    Let $G=(V, E, w)$ be a weighted \textit{undirected} graph.
    Let $k\ge 1$ be an integer. There is a $(2k-1)$-stretch compact roundtrip routing scheme that uses local routing tables of size $\tilde{O}(n^{1/k})$, vertex labels of size $\tilde{O}(k)$ and packet headers of size $\tilde{O}(k)$. 
}
\end{theorem}
Using this result with the handshake-based routing scheme of Thorup and Zwick~\cite{DBLP:conf/spaa/ThorupZ01}, one obtains an optimal $(2k-1)$-stretch compact routing scheme for any communication session.
We summarize our new results for undirected graphs and compare them to the previous work in~\autoref{tab:compact-undirected}.

\begin{table}[t]
    \centering
    \begin{tabular}{|c|c|c|c|c|}
        \hline
         Stretch & Local storage & \makecell{Uses average\\local storage?} &  Ref. & Comments\\ \hline\hline
         $O(k)$ &  $\Ot(n^{1/k})$ & yes & \cite{DBLP:journals/jacm/PelegU89} & unweighted graphs \\  \hline
         $2^k-1$ &  $\Ot(n^{1/k})$ & yes & \cite{DBLP:journals/jal/AwerbuchBLP90} & \\  \hline
         $O(k^2)$ &  $\Ot(n^{1/k})$ & no & \cite{DBLP:journals/siamdm/AwerbuchP92} & \\  \hline
         $4k-5$ & $\Ot(n^{1/k})$ & no & \cite{DBLP:conf/spaa/ThorupZ01} & \\  \hline
         $\approx 3.68k$ & $\Ot(n^{1/k})$ & no & \cite{DBLP:conf/podc/Chechik13} & \\  \hline
         $4k-7+\eps$ & $\Ot(n^{1/k})$ & no & \cite{DBLP:conf/podc/RodittyT15} & \\  \hline
         $\approx 2.64k$ & $\Ot(n^{1/k})$ & yes &\autoref{T-Compact-Amortized-3k} & \\  \hline
         $2k-1$ & $\Ot(n^{1/k})$ & no &\autoref{T-Compact-Undirected} & roundtrip stretch \\ \hline
    \end{tabular}
    \caption{Compact routing schemes in undirected graphs.}
    \label{tab:compact-undirected}
\end{table}

Next, we turn our attention to weighted \textit{directed} graphs. 
Since preserving the asymmetric reachability structure of directed graphs requires $\Omega(n^{2})$ space\footnote{In a bipartite graph in which all the edges are from one side to the other, there are $\Theta(n^2)$ edges and removing each of them makes the destination not reachable from the source. Thus, storing reachability requires $\Omega(n^2)$ space.}, no spanner, emulator, or compact routing scheme can exist for directed graphs. 
Cowen and Wagner~\cite{DBLP:conf/soda/CowenW99} circumvented the $\Omega(n^2)$ space lower bound by introducing roundtrip distances in directed graphs, defined as $d(u \leftrightarrow v) = d(u,v) + d(v,u)$. 

In the last few decades, a few compact roundtrip routing schemes were presented, see for example~\cite{DBLP:conf/soda/CowenW99, DBLP:journals/jal/CowenW04, DBLP:journals/talg/RodittyTZ08}. The state-of-the-art result was obtained by Roditty, Thorup, and Zwick~\cite{DBLP:journals/talg/RodittyTZ08}. They presented a $3$-stretch compact roundtrip routing scheme that uses $\Ot(n^{1/2})$ local storage and also a $(4k+\eps)$-stretch compact roundtrip routing scheme that uses $\Ot(n^{1/k})$ local storage for every $k\ge 3$. 

From the Erd\H{o}s girth conjecture, it follows that every compact roundtrip routing scheme with stretch $< 2k + 1$ must use total storage of $\Omega(n^{1+1/k})$ bits. Closing the gap between the upper and the lower bound is the main open problem regarding compact roundtrip routing schemes.

\begin{problem}\label{Pr-2}
    For every $k\ge 2$, given a directed weighted graph, what is the best stretch of a roundtrip routing scheme that uses $\Ot(n^{1/k})$ local storage?
\end{problem}

In recent years, roundtrip distances have been extensively studied but only in the context of roundtrip spanners and roundtrip emulators (see, for example~\cite{DBLP:conf/soda/HarbuzovaJWX24, CenDG20, DBLP:journals/talg/RodittyTZ08, DBLP:conf/cocoon/StaffordZ23, DBLP:journals/corr/abs-1911-12411, DBLP:conf/soda/PachockiRSTW18}).
Despite all the recent progress, no improvements were obtained for compact roundtrip routing schemes since the $(4k+\eps)$-stretch roundtrip routing scheme of~\cite{DBLP:journals/talg/RodittyTZ08} from more than two decades ago.  

In this paper, we improve upon~\cite{DBLP:journals/talg/RodittyTZ08} for the case that $k=3$. More specifically, 
using $\Ot(n^{1/3})$ local storage, Roditty, Thorup, and Zwick~\cite{DBLP:journals/talg/RodittyTZ08} obtained a  $(12+\eps)$-stretch  roundtrip routing scheme. We present a $7$-stretch roundtrip routing scheme that uses $\Ot(n^{1/3})$ local storage, as presented in the following theorem.
\begin{theorem} \label{T-Compact-Directed-7-1/3} \Copy{T-Compact-Directed-7-1/3} {
    Let $G=(V, E, w)$ be a weighted \textit{directed} graph.
    There is a $7$-stretch compact roundtrip routing scheme that uses local routing tables of size $\tilde{O}(n^{1/3})$, vertex labels of size $\tilde{O}(1)$ and packet headers of size $\tilde{O}(1)$. }
\end{theorem}

In addition, in the following theorem, we present an optimal\footnote{Assuming the Erd\H{o}s conjecture, it is easy to create a graph $G$ with $\Omega(n^{1+1/k})$ edges such that the girth of $G$ is $2k+2$, and the diameter of $G$ is at most $2k+1$.
Given a graph with  $\Omega(n^{1+1/k})$ edges and $2k+2$ girth, if there is a pair of vertices $u,v\in V$ such that $d(u,v) \ge 2k+2$, we can add an edge between $u$ and $v$ to the graph without reducing the girth. By repeating this process until there are no pairs $u,v\in V$ such that $d(u,v) \ge 2k+2$, we get that the diameter is at most $2k+1$. }, up to polylogarithmic factors, 
compact roundtrip routing scheme in graphs with $D_{hop} = \Ot(k)$, where $D_{hop}$ is the hop diameter of the graph.
\begin{theorem} \label{T-Compact-directed-bounded-diameter} \Copy{T-Compact-directed-bounded-diameter}{
    Let $G=(V, E, w)$ be a weighted \textit{directed} graph.
    Let $k\ge 1$ be an integer. There is a $(2k-1)$-stretch compact roundtrip routing scheme that uses local routing tables of size $\tilde{O}(D_{hop} n^{1/k})$, vertex labels of size $\tilde{O}(D_{hop}k)$ and packet headers of size $\tilde{O}(D_{hop})$. }
\end{theorem}
We summarize our new results for directed graphs and compare them to the previous work in~\autoref{tab:roundtrip}.
\begin{table}[t]
    \centering
    \begin{tabular}{|c|c|c|c|c|}
         \hline
         Stretch & local storage & \makecell{Uses average\\local storage?} & Ref. & Comments \\ \hline\hline
         $3$ & $\Ot(n^{1/2})$ & yes & \cite{DBLP:conf/soda/CowenW99} & \\ \hline
         $3$ & $\Ot(n^{2/3})$ & no & \cite{DBLP:conf/soda/CowenW99} & \\ \hline
         $2^k-1$ & $\Ot(n^{1/k})$ & yes & \cite{DBLP:journals/jal/CowenW04} & \\ \hline
         $3$ & $\Ot(n^{1/2})$ & no & \cite{DBLP:journals/talg/RodittyTZ08} & \\ \hline
         $4k+\eps$ & $\Ot(\frac{1}{\eps}n^{1/k})$ & no & \cite{DBLP:journals/talg/RodittyTZ08} & $\eps > 0$\\ \hline
         $12+\eps$ & $\Ot(\frac{1}{\eps}n^{1/3})$ & no & \cite{DBLP:journals/talg/RodittyTZ08} & $\eps > 0$\\ \hline
         $7$ & $\Ot(n^{1/3})$ & no & \autoref{T-Compact-Directed-7-1/3}&  \\ \hline
         $2k-1$ & $\Ot(D_{hop}\cdot n^{1/k})$ & no &~\autoref{T-Compact-directed-bounded-diameter} & $D_{hop}$ is the hop diameter of $G$. \\ \hline
    \end{tabular}
    \caption{Compact roundtrip routing results in directed weighted graphs.\protect\footnotemark}
    \label{tab:roundtrip}
\end{table}
\footnotetext{poly-logarithmic factors are omitted.}

The rest of this paper is organized as follows. 
In~\autoref{S-Prel} we present some necessary preliminaries. 
In~\autoref{S-Over} we present an overview of our technical contributions and our new compact routing schemes. 
In~\autoref{S-roundtrip-undirected} we present our optimal compact roundtrip routing scheme for weighted undirected graphs.
In~\autoref{S-roundtrip directed} we present our two new compact roundtrip routing schemes for weighted directed graphs.
Finally, in~\autoref{S-amortized-undirected} we present our single message compact routing scheme for weighted undirected graphs that use average local storage.

\section{Preliminaries}\label{S-Prel}
Let $G = (V, E, w)$ be a weighted graph with $n$ vertices and $m$ edges, where $w: E \to \mathbb{R}^+$.
We consider both connected undirected graphs and strongly connected directed graphs\footnote{
If the graph is not connected (or not strongly connected), we add a dummy vertex and connect it with bi-directional edges of weight $\infty$ to every vertex of the graph.}.

Let the distance $d_G(u,v)$ from $u$ to $v$ be the length of the shortest path from $u$ to $v$ in $G$, where the length of a path is the sum of its edge weights, and let $P_G(u,v)$ be the shortest path from $u$ to $v$, $G$ is omitted when it is clear from context. The roundtrip distance $d(u \leftrightarrow v)$ is defined as $d(u,v) + d(v,u)$. Throughout this paper, we assume that for any two vertices $u$ and $v$, there exists a unique shortest path between them. In the case of multiple shortest paths of the same length, we break ties by selecting the path with the lexicographically smallest sequence of vertex identifiers.

Let $X \subseteq V$. The distance $d(u,X)$ from $u$ to $X$ is the distance between $u$ and the closest vertex to $u$ from $X$, that is, 
$d(u,X)= \min_{x\in X}(d(u,x))$. Similarly, the roundtrip distance from $u$ to $X$ is defined as $d(u\leftrightarrow X) = \min_{x\in X}(d(u\leftrightarrow x))$.
Let $p(u, X)=\arg \min_{x\in X}(d(u\leftrightarrow x))$ (ties are broken in favor of the vertex with a smaller identifier). 

Next, following the ideas of Thorup and Zwick~\cite{DBLP:journals/jacm/ThorupZ05}, we define bunches and clusters. 
Let $u\in V$ and let $X, Y\subseteq V$. 
The bunch of $u$ with respect to $X$ and $Y$ is defined as $B(u,X,Y)=\{ v\in X \mid d(u \leftrightarrow v) < d(u \leftrightarrow Y)\}$. The ball of $u$ with respect to $Y$ is defined as $B(u,Y)=\{ v\in V \mid d(u \leftrightarrow v) < d(u \leftrightarrow Y)\}$ (notice that $B(u,Y)=B(u,V,Y)$).
The cluster of $u$ with respect to $Y$ is defined as $C(u,Y)=\{ v\in V \mid d(u \leftrightarrow v)<d(v \leftrightarrow Y)\}$. 

The starting point in many algorithms, distance oracles, and compact routing schemes, and in particular in Thorup and Zwick~\cite{DBLP:conf/spaa/ThorupZ01} routing scheme, is a hierarchy of vertex sets  $A_0, A_1,\ldots, A_k$, where $A_0=V$, $A_k=\emptyset$, $A_{i+1}\subseteq A_i$ and $|A_{i}|=n^{1-i/k}$ for $0 \leq i \leq k-1$.

For every $0 \leq i \leq k-1$, the $i$-th pivot of $u$ is defined as 
$p_i(u) = p(u, A_i)$, and $h_i(u)$ is defined as $d(u,A_i)$.
The $i$-th bunch of $u$ is defined as 
$B_i(u) = B(u, A_i, A_{i+1})$.
The bunch of $u$ is defined as the union of its individual bunches, that is, 
$B(u) = \bigcup_{i=0}^{k-1} B_i(u)$.
The cluster of a vertex $w \in A_i \setminus A_{i+1}$ is defined as the cluster of $w$ with respect to the set $A_{i+1}$, that is, 
$C(w) = C(w, A_{i+1})$.
We denote by $[k]$ the set $\{0,1,2,\ldots,k-1\}$.

In the following lemma, we provide an upper bound for the size of $B(u)$, which is $O(kn^{1/k})$, as demonstrated by Thorup and Zwick~\cite{DBLP:journals/jacm/ThorupZ05}.
\begin{lemma}[\cite{DBLP:journals/jacm/ThorupZ05, DBLP:conf/icalp/RodittyTZ05}]\label{L-TZ-Size}
    Given an integer parameter $k\geq 2$, we can compute
    sets $A_1,\ldots, A_{k-1}$, such that $|A_i|=O(n^{1-i/k})$, for every $1\le i \le k-1$.
    For every $i\in [k]$ the size of $B_i(u)$ is $\Ot(n^{1/k})$.
\end{lemma}

In the following lemma, we provide an upper bound for the size of $C(u)$, which is $O(n^{1/k})$, as demonstrated by Thorup and Zwick~\cite{DBLP:conf/spaa/ThorupZ01}.
\begin{lemma}[\cite{DBLP:conf/spaa/ThorupZ01}]\label{L-A-center}
Given a parameter $p$, we can compute a set $A$ of size $\Ot(np)$ 
such that, $|C(w,A)|=O(1/p)$, for every vertex $w\in V \setminus A$, and $|B(v,V,A)|=O(1/p)$ for every $v\in V$. 
\end{lemma}

Let $S \subseteq V$. We define $T_{out}(u, S)$ as a tree containing the directed shortest paths from $u$ to all the vertices in $S$, and $T_{in}(u, S)$ as a tree containing the directed shortest paths from all the vertices in $S$ to $u$. When $u$ is clear from context, we omit it, for example, we use $T(C(u))=T(u, C(u))$, and $T(B(u))=T(u,B(u))$. 
Note that it is possible for $|T_{out}(u, S)|$ to be bigger than $|S|$ in cases where the shortest path from $u$ to a vertex in $S$ passes through a vertex not in $S$. \footnote{We denote with $|H|$ the number of edges in $H$, i.e. $|H|=|E(H)|$.}
Let $T(u, X) = T_{in}(u, X) \cup T_{out}(u, X)$, as defined in~\cite{DBLP:journals/talg/RodittyTZ08}, when $u$ is clear from context we omit it and write $T(X)$. Notice that in undirected graphs, since $T_{in}(u,X) = T_{out}(u,X)$, we have that $T(u,X) = T_{in}(u,X) = T_{out}(u,X)$.
Next, we show that if $S$ is a ball, i.e. $S = B(u, X) = B(u,V,X)$ for some set $X$, then $|T_{out}(u,S)| = |S|$ and $|T_{in}(u,S)| = |S|$, and therefore $|T(u,S)| \le 2|S|$.

\begin{lemma}[\cite{DBLP:journals/talg/RodittyTZ08}]
    \label{L-Ball-can-be-tree}
    $|T_{out}(u,B(u,X))| = |B(u,X)|$ and $|T_{in}(u,B(u,X))| = |B(u,X)|$.
\end{lemma}
\begin{proof}
Let $u \in V$, $v \in B(u, X)$, and let $w \in P(u,v)$. We will show that $w \in B(u, X)$. Since all the vertices in $T_{out}$ are on the shortest path from $u$ to a vertex in $B(u, X)$, we obtain that $|T_{out}(u, B(u, X))| \le |B(u, X)|$, as required.
The proof for the in-ball is identical for reversed paths.
From the triangle inequality, we know that 
\[
d(u \leftrightarrow w) = d(u, w) + d(w, u) \stactri\le d(u, w) + d(w, v) + d(v, u) = d(u \leftrightarrow v) < d(u\leftrightarrow X),
\]
where the last inequality follows from the fact that $v \in B(u, X)$.\footnote{$\stactri \le$ denotes an inequality that follows from the triangle inequality.}
Since $d(u \leftrightarrow w) < d(u\leftrightarrow X)$ we get that $w\in B(u,X)$, as required.
\end{proof}

The following lemma was originally proven in~\cite{DBLP:journals/jacm/ThorupZ05} for any metric space and therefore holds also for roundtrip distances. For completeness, we prove the lemma for roundtrip distances.
\begin{lemma} \label{L-Bound-p-i-d}
     Let $u,v\in V$ and let $0 < i \leq k-1$. If $p_j(u) \notin B(v)$ and $p_j(v) \notin B(u)$ for every $0 < j < i$, then
\[
d(u \leftrightarrow p_i(u)) \leq i \cdot d(u \leftrightarrow v) \quad \text{and} \quad d(v \leftrightarrow p_i(v)) \leq i \cdot d(u \leftrightarrow v).
\]
\end{lemma}
\begin{proof}
We prove the claim by induction for every $0\leq i\leq \ell$. For $i = 0$, the lemma holds since $d(v \leftrightarrow v) = 0$ and $d(u \leftrightarrow u) = 0$.
Next, we prove the induction step. We assume the correctness of the claim for $i - 1$ and show its correctness for $i$. Therefore, $d(v \leftrightarrow p_{i - 1}(v)) \leq (i - 1) \cdot d(u \leftrightarrow v)$ and $d(u \leftrightarrow p_{i - 1}(u)) \leq (i - 1) \cdot d(u \leftrightarrow v)$. Without loss of generality, we show that $d(u \leftrightarrow p_i(u)) \leq i \cdot d(u \leftrightarrow v)$. The proof for $v$ is identical.

Since $i \leq \ell$, it follows that $i - 1 < \ell$. Therefore, from the lemma's assumptions, we know that $p_{i - 1}(v) \notin B(u)$. From the definition of $B(u)$, it follows that 
$d(u \leftrightarrow p_i(u)) \leq d(u \leftrightarrow p_{i - 1}(v))$.
From the triangle inequality, it follows that 
$d(u \leftrightarrow p_{i - 1}(v)) \leq d(u \leftrightarrow v) + d(v \leftrightarrow p_{i - 1}(v))$. Recall that from the induction assumption we know that  
$d(v \leftrightarrow p_{i - 1}(v)) \leq (i - 1) \cdot d(u \leftrightarrow v)$. Therefore, we get that:
\[
d(u \leftrightarrow p_i(u)) \leq d(u \leftrightarrow p_{i - 1}(v)) \leq d(u \leftrightarrow v) + d(v \leftrightarrow p_{i - 1}(v)) \leq d(u \leftrightarrow v) + (i - 1) \cdot d(u \leftrightarrow v) = i \cdot d(u \leftrightarrow v),
\]
as required.
\end{proof}

Next, we show that if $p_{i - 1}(v) \notin B(u)$, then 
$d(v \leftrightarrow p_i(v)) \leq d(v \leftrightarrow p_{i - 1}(v)) + 2 \cdot d(u \leftrightarrow v)$.
\begin{lemma}
    \label{L-Bound-delta-2d}
    Let $u,v\in V$, if $p_{i-1}(v)\notin B(u)$ then $d(v\leftrightarrow p_{i}(v)) \le d(v\leftrightarrow p_{i-1}(v)) + 2d(u \leftrightarrow v)$
\end{lemma}
\begin{proof}
    From the definition of $p_i(v)$, we know that it is the closest vertex to $v$ in $A_i$, and since $p_i(u)\in A_i$ we get that  
    $d(v \leftrightarrow p_i(v)) \leq d(v \leftrightarrow p_i(u))$.
    From the triangle inequality, it follows that
    $d(v \leftrightarrow p_i(u)) \leq d(v \leftrightarrow u) + d(u \leftrightarrow p_i(u))$.
    From the lemma assumption, we know that $p_{i - 1}(v) \notin B(u)$. Therefore,    $d(u \leftrightarrow p_i(u)) \leq d(u \leftrightarrow p_{i - 1}(v))$.
    From the triangle inequality, it follows that
    $d(u \leftrightarrow p_{i - 1}(v)) \leq d(u \leftrightarrow v) + d(v\leftrightarrow p_{i - 1}(v))$.
    Overall, we get that:
    \begin{align*}
        d(v \leftrightarrow p_i(v)) &\leq d(v \leftrightarrow p_i(u)) \leq d(v \leftrightarrow u) + d(u \leftrightarrow p_i(u)) \leq d(v \leftrightarrow u) + d(u \leftrightarrow p_{i - 1}(v)) 
        \\&\leq d(v \leftrightarrow u) + d(u \leftrightarrow v) + d(v\leftrightarrow p_{i - 1}(v)) = 2d(u \leftrightarrow v) + d(v, p_{i - 1}(v)),
    \end{align*}
    as required.
\end{proof}

The following lemma holds only for \textit{undirected} graphs and is presented in~\cite{DBLP:journals/jacm/ThorupZ05}.
\begin{lemma}[\cite{DBLP:journals/jacm/ThorupZ05}]\label{L-T(C(u))=C(u)-undirected}
    Let $G=(V, E, w)$ be a weighted \textit{undirected} graph.
    Let $v\in A_i\setminus A_{i+1}$, let $u \in C(v)$, and let $w\in P(v, u)$ then $w\in C(v)$.
\end{lemma}
\begin{proof}
    For the sake of contradiction, assume that $w \notin C(v)$. From the definition of $C(v)$, we have that $d(v\leftrightarrow w) \ge d(w\leftrightarrow A_{i+1})$. 
    
    By the triangle inequality, we have:
    $
    d(u\leftrightarrow A_{i+1}) \le d(u\leftrightarrow w) + d(w\leftrightarrow A_{i+1}).
    $
    Since $d(w\leftrightarrow A_{i+1}) \le d(v\leftrightarrow w)$, we get $
    d(u\leftrightarrow A_{i+1}) \le d(u\leftrightarrow w) + d(w\leftrightarrow A_{i+1}) \le d(u\leftrightarrow w) + d(v\leftrightarrow w)$.
    Since $w\in P(v,u)$, and the graph is undirected, it follows that $d(u\leftrightarrow w) + d(v\leftrightarrow w) = d(v\leftrightarrow u)$.
    Thus, 
    $
    d(u\leftrightarrow A_{i+1}) \le d(v\leftrightarrow u),
    $
    a contradiction to the fact that $u \in C(v)$.
\end{proof}

\subsection{General framework}
A routing scheme comprises two phases: a preprocessing phase and a routing phase.
In the preprocessing phase, the entire graph is accessible to the algorithm. The preprocessing algorithm computes for every $u\in V$ a routing table $\RT(u)$ and a label $\L(u)$. 
Each vertex $u$ saves $\RT(u)$  and $\L(u)$ in its local storage.
A routing scheme is considered \emph{compact} if the size of the routing tables is sub-linear in the number of vertices, i.e., $|\RT(u)| = o(n)$.

In the routing phase, the goal is to route a message from a source vertex $u$ to a destination vertex $v$. 
Specifically, the routing algorithm at the source vertex $u$ gets as input 
the routing table $\RT(u)$, and the labels $\L(u)$ and $\L(v)$.  Based on this input, the routing algorithm determines a neighboring vertex of $u$ to which the message should be forwarded. The routing algorithm can also attach a header to the message. 
When a vertex $w$ receives  a message, 
the routing algorithm at $w$ gets as input the routing table 
$\RT(w)$, and the labels $\L(w)$ and $\L(v)$ (and possibly a header). Based on this input, the routing algorithm determines a neighboring vertex of $w$ to which the message should be forwarded. 
The message is routed from a vertex to one of its neighbors until the message reaches its destination vertex $v$.

We denote by $\hat{d}(u,v)$ the distance that a message whose source vertex is  $u$ and whose destination vertex is $v$ traverses from $u$ to $v$. 
 The stretch of the routing scheme is defined as $\max_{u,v\in V}(\frac{\hat{d}(u,v)}{d(u,v)})$.

Several variants of compact routing schemes exist. 
In a \textit{labeled} routing scheme, the preprocessing algorithm can assign labels to the vertices. In a \textit{fixed-port} routing scheme, 
the order of the neighbors of each vertex is fixed and cannot be changed by 
the preprocessing algorithm.
In this work, we focus on labeled, fixed-port compact routing schemes.  

\subsection{Routing in trees}
An essential ingredient in our compact routing schemes for general graphs is the following compact routing scheme for trees.
Given a tree $T$, the preprocessing algorithm assigns a label $\L(v, T)$ to every vertex $v$ in $T$. The routing algorithm then routes a message from a source vertex $u$ to a destination vertex $v$ on the shortest path from $u$ to $v$ in $T$.
Thorup and Zwick~\cite{DBLP:conf/spaa/ThorupZ01} presented a tree routing scheme that uses only vertex labels of size $(1 + o(1))\log n$ and no routing tables.
In the fixed-port model, however, the label size increases to $O(\log^2 n)$.
A similar scheme was presented by Fraigniaud and Gavoille~\cite{DBLP:conf/stacs/FraigniaudG02}.
The following lemma outlines the known results for tree compact routing schemes.
\begin{lemma}[\cite{DBLP:conf/stacs/FraigniaudG02, DBLP:conf/spaa/ThorupZ01}]
    \label{T-Route-In-Trees}
    Let $T=(V, E)$ be an undirected tree on $n$ vertices with each edge $e\in E$ assigned a unique $O(\log n)$-bit port number. Then, it is possible to efficiently assign each vertex $v\in V$ an $O(\log^2n/\log\log n)$-bit label, denoted $label(v)$, such that if $u,v\in V$, then given $label(u)$ and $label(v)$, and nothing else, it is possible to find in constant time, the port number assigned to the first edge on the path in $T$ from $u$ to $v$.
\end{lemma}

\subsection{Routing in directed graphs}
Roditty, Thorup, and Zwick~\cite{DBLP:journals/talg/RodittyTZ08} extended the compact routing scheme from trees to double trees to handle directed graphs. In our work on directed graphs (\autoref{S-roundtrip directed}), we employ this double-tree adaptation when routing within double-trees.
Moreover, for the routing in clusters to work in directed graphs, they adjusted the definition of cluster adaptation using the following definition of roundtrip ordering. 
\begin{definition}
    We assume that $V=\{1,\dots,n\}$. Let $u_1,u_2,v\in V$. We say that $u_1 \prec_v u_2$ if one of the following holds:
    \begin{itemize}
        \item $d(v\leftrightarrow u_1) < d(v\leftrightarrow u_1)$
        \item $d(v\leftrightarrow u_1) = d(v\leftrightarrow u_1)$ and $d(u_1\rightarrow v) < d(u_1\rightarrow v)$
        \item $d(v\leftrightarrow u_1) = d(v\leftrightarrow u_1)$ and $d(u_1\rightarrow v) = d(u_1\rightarrow v)$ and $u_1< u_2$
    \end{itemize}
\end{definition}
Using this definition, the cluster of $u$ with respect to $Y$ is defined as $C(u,Y)=\{ v\in V \mid u\nprec_v p_Y(u)\}$, where $p_Y(u)$ satisfies that $p_Y(u) \preceq_v w$ for every $w\in Y$.

Using this definition, they proved the following lemma:
\begin{lemma}[\cite{DBLP:journals/talg/RodittyTZ08}]\label{L-Digraph-Clusters-are-continuent}
    If $v\in C(u)$ and $w\in P(u,v)$, then $v\in C(w)$
\end{lemma}

\section{Overview} \label{S-Over}
In this section, we present an overview of our new compact routing schemes and their main technical contributions.
Throughout the overview, let $V=A_0, \dots, A_k=\emptyset$ be a hierarchy such that, $|A_i|=\Ot(n^{1-i/k})$ and $|B(u)|=\Ot(n^{1/k})$, for every $u\in V$, created using~\autoref{L-TZ-Size}. We denote with  $x\rightarrow y$ routing from $x$ to $y$ on $P(x,y)$. 

\textbf{Roundtrip routing scheme in undirected graphs.} 
Notice that any $t$-stretch compact routing scheme is also a $t$-roundtrip stretch compact routing scheme.
In undirected graphs, where  $d(u,v) = d(v,u)$, we have $d(u \leftrightarrow v)=d(u,v)+d(v,u)=2d(u,v)$.
This might lead one to question the potential benefits of considering roundtrip routing in undirected graphs.
In other words, why might the problem of roundtrip routing be easier than the problem of single message routing, even though $d(u \leftrightarrow v) = 2d(u,v)$?

During the routing process, the information available to the routing algorithm at $u$ may differ from the information available at $v$. Therefore,  the routing path from $u$ to $v$ may differ from the routing path from $v$ to $u$,  which might lead to the case that  $\hat{d}(u,v) \neq \hat{d}(v,u)$.
Consider for example the case that  $\hat{d}(u,v)=3d(u,v)$ and $\hat{d}(v,u)=27d(u,v)$. In this case, the roundtrip stretch is $15$ while the single message stretch is $27$, and even though $d(u \leftrightarrow v)=2d(u,v)$,  the roundtrip stretch is much smaller than the single message stretch.

Next, we provide an overview of our optimal $(2k-1)$-roundtrip stretch compact routing scheme (for the complete description, see~\autoref{S-roundtrip-undirected}). 
The preprocessing algorithm sets $\RT(u)=\{ \L(u, T(C(v))) \mid v \in B(u) \}$ and $\L(u)=\{ \L(u, T(C(p_i(u)))) \mid i \in [k] \}$, for every $u\in V$.
Let $\ell(x,y)=\min\{i\mid p_{i}(y)\in B(x)\}$, and let $b = \ell(u, v)$, and let $a = \ell(v, u)$. The roundtrip routing path is $u \rightarrow p_{b}(v) \rightarrow v \rightarrow p_a(u) \rightarrow u$ (see Figure~\ref{fig:undirected-roundtrip}). 

Our main technical contribution is in Lemma~\ref{L-Routing-undirected-Approximation}, where we show  that while the path 
$u\rightarrow p_b(v) \rightarrow v$ might be of length at most $(4k-3)d(u,v)$, the entire path $u \rightarrow p_{b}(v) \rightarrow v \rightarrow p_a(u) \rightarrow u$ is of length of at most $(4k-2)d(u,v)=(2k-1)d(u \leftrightarrow v)$, and therefore the roundtrip stretch is the optimal $(2k-1)$-stretch. 

Next, we provide some intuition why $\hat{d}(u \leftrightarrow v) \le (2k - 1) d(u \leftrightarrow v)$.
From the triangle inequality, we have that 
$\hat{d}(u\leftrightarrow v) \le d(u \leftrightarrow v) + d(u\leftrightarrow p_a(u)) + d(v\leftrightarrow p_b(v))$. 
From the definition of $a$ and $b$,  for every $0 < i < a \le b\leq k-1$ we have  $p_i(u) \notin B(v)$ and $p_i(v) \notin B(u)$, and for every $a \le i < b$ we have  $p_i(v) \notin B(u)$. This allows us to prove that $d(u\leftrightarrow p_a(u)) \le a\cdot d(u \leftrightarrow v)$ and $d(v\leftrightarrow p_b(v)) \le a\cdot d(u \leftrightarrow v) + (b-a)\cdot2d(v\leftrightarrow p_b(v))$.
Overall: 
\begin{align*}
    d(u\leftrightarrow p_a(u)) + d(v\leftrightarrow p_b(v)) &\le ad(u \leftrightarrow v) + a d(u \leftrightarrow v) + (b-a)2d(u \leftrightarrow v) 
    \\&= 2bd(u \leftrightarrow v) \le 2(k-1)d(u \leftrightarrow v)
\end{align*}
and therefore $\hat{d}(u\leftrightarrow v) \le d(u \leftrightarrow v) + d(u\leftrightarrow p_a(u)) + d(v\leftrightarrow p_b(v)) \le (2k-1)d(u \leftrightarrow v)$.

\textbf{Directed roundtrip routing schemes.}
One might wonder why the general compact roundtrip routing scheme presented above for undirected graphs can not be extended to directed graphs. 
The problem lies in the structure of clusters. In undirected graphs, if $u\in C(w)$ then $P(u,w)\subseteq C(w)$ (see~\autoref{L-T(C(u))=C(u)-undirected}), and therefore we can route from $u$ to every $w$, such that  $u\in C(w)$. Unfortunately, in directed graphs,  
this nice property of clusters does not necessarily hold. More specifically, it might be that $P(u,w)\not\subseteq C(w)$ even when $u\in C(w)$
and as a result, we cannot route from $u$ to $w$ while using only the cluster $C(w)$ as in undirected graphs. 
A simple approach to overcome this problem is to store $P(u,w)$, for every $w\in B(u)$, in  $\RT(u)$.
Using this approach, we can extend the roundtrip routing scheme from undirected graphs to directed graphs with small hop diameter (see~\autoref{T-Compact-directed-bounded-diameter}). 

For routing tables of size $\Ot(n^{1/3})$, we overcome the problem that  $P(u,w)\not\subseteq C(w)$ using a more sophisticated solution.
For an hierarchy with $k=3$ we have  $V=A_0,A_1,A_2,A_3=\emptyset$ and three bunches, $B_0(u),B_1(u)$ and $B_2(u)$, for every $u\in V$.
Let $u,v \in V$. To follow the compact roundtrip of undirected graphs we need to be able to route from $u$ to $v$, if $v\in B_0(u)$, from $u$ to $p_1(v)$, if $p_1(v)\in B_1(u)$, and from $u$ to $p_2(v)$, otherwise.
Since  $B_0(u)$ is a ball $P(u,v)\subseteq B_0(u)$, for every $v\in B_0(u)$ (see \autoref{L-Ball-can-be-tree}). Therefore,  we can route from $u$ to every vertex in $B_0(u)$.
Moreover, since $C(p_2(v))=V$, we can route from $u$ to $p_2(v)$.
This leaves us with the challenge of routing from $u$ to $p_1(v)$ when $p_1(v)\in B_1(u)$. 

To handle this challenge, we divide the routing from $u$ to $p_1(v)$ into two cases. 
If $P(u,p_1(v))\subseteq C(p_1(v))$, we simply route on the path $P(u,p_1(v))$. Otherwise, if $P(u,p_1(v))\not\subseteq C(p_1(v))$, we let $z \in P(u,p_1(v))$ be the first vertex such that $z\notin C(p_1(v))$. In this case, we route along the path $u\rightarrow p_2(z) \rightarrow p_1(v)$. 
Using the fact that $z\notin C(p_1(v))$ we show that $d(u, p_2(z)) + d(p_2(z), p_1(v)) \le d(u,p_1(v)) + d(u \leftrightarrow p_1(v))$,
and that $\hat{d}(u\leftrightarrow v) \le 7d(u \leftrightarrow v)$ (see~\autoref{L-Routing-directed-Approximation}).

\textbf{Average single message routing schemes in undirected graphs.}
Recall that $h_i(u)=d(u,A_i)$. In the preprocessing algorithm, we set $\RT(u)=\{ \L(u, T(C(v))) \mid v \in B(u) \} \cup \{L(v,T(C(u))) \mid v\in C(u) \}$, and $\L(u) = \{ \pair{h_i(u), \L(u, T(C(p_i(u))))} \mid i \in [k] \}$, for every $u\in V$.

For the sake of simplicity, we assume that when routing from the source vertex $u$ to the destination vertex $v$, the value $d(u,v)$ is known to the routing algorithm. In this case, we present an optimal $(2k-1)$-stretch routing scheme in Section~\ref{S-average-simple}.
The algorithm at the source $u$ works as follows.
For each $0 \le i \le k-1$, if $p_i(v) \in B_i(u)$, it routes from $u$ to $v$ in $T(C(p_i(v)))$. Otherwise, if the inequality $h_{i+1}(v) > d(u,v) + h_i(u)$ holds, then we have that $v\in C(p_i(u))$ (see~\autoref{L-average-simple}) and therefore the algorithm routes from $u$ to $v$ in $T(C(p_i(u)))$.
If neither condition is satisfied, the algorithm proceeds to the next iteration. 
The $(2k-1)$-stretch of the routing scheme follows by induction using standard tools (see~\autoref{L-simple-average-corr} for the complete proof).

In Section~\ref{S-average-full}, we present a routing algorithm that routes from $u$ to $v$ without knowing the value of $d(u,v)$.
To achieve this, we introduce an estimate $\hat{\delta}$ satisfying $\hat{\delta} \le d(u,v)$, which serves as our current best lower bound for $d(u,v)$.
Note that since $\hat{\delta}$ is only an estimate and may be strictly smaller than $d(u,v)$, it is possible that $h_{2i+1}(v) > \hat{\delta} + h_{2i}(u)$ and $v \notin C(p_{2i}(u))$.\footnote{
We alternate between bounding $h_i(u)$ and $h_i(v)$, but by the definition of $\hat{\delta}$, we have that $h_{i+1}(u) \le \hat{\delta} + h_i(v)$ for every $i \le \ell$.  
Therefore, it suffices to only bound the $h_{2i+1}(v) - h_{2i}(u)$ for $0\le i\le a$.
}
Therefore, we introduce a condition that, if satisfied, means that $\hat{\delta} \ll h_{2i+1}(v) - h_{2i}(u)$.
If the condition is satisfied, the routing algorithm routes to $p_{2i}(u)$ and checks in $\RT(p_{2i}(u))$ whether $v\in C(p_{2i}(u))$. If $v\in C(p_{2i}(u))$ the algorithm simply routes from $p_{2i}(u)$ to $v$. Otherwise, if $v\notin C(p_{2i}(u))$, then we know that $ h_{2i+1}(v) - h_{2i}(u) \le d(u,v)$, and we can safely update $\hat{\delta}$ to $h_{2i+1}(v) - h_{2i}(u)$. We proceed by routing back from $p_{2i}(u)$ to $u$ and then continuing to the next iteration of the algorithm.

The routing algorithm is simple and works as follows. Let $\ell = \min\{i \mid p_{i}(v) \in B(u)\}$, and let $\hat{\delta} = \max_{0 \leq i \leq \ell}(h_{i+1}(u) - h_i(v))$. For each $0\le i \le \ell/2$, let $1<c_i<2$ be a constant:
\begin{enumerate}
    \item If $h_{2i+1}(v) \le c_i \cdot \hat{\delta} + h_{2i}(u)$, then the algorithm continues to the next iteration.
    \item Otherwise, if $h_{2i+1}(v) > c_i \cdot \hat{\delta} + h_{2i}(u)$, then the algorithm routes to $p_{2i}(u)$ and accesses $\RT(p_{2i}(u))$:
    \begin{enumerate}
        \item If $v\in C(p_{2i}(u))$ then the algorithm routes directly from $p_{2i}(u)$ to $v$ in $T(C(p_{2i}(u)))$.
        \item Otherwise, if $v\notin C(p_{2i}(u))$, the algorithm sets $\hat{\delta}$ to $h_{2i+1}(v) - h_{2i}(u)$ and routes back to $u$ from $p_{2i}(u)$. The algorithm then continues to the next iteration.
    \end{enumerate}
\end{enumerate}

In contrast to the simplicity of the routing algorithm, analyzing its stretch is rather involved.  
For a complete proof, see~\autoref{L-stretch-undirected-average}.



\section{Optimal roundtrip routing in undirected graphs}\label{S-roundtrip-undirected}
In this section, we consider roundtrip routing in weighted undirected graphs. In undirected graphs $d(u,v) = d(v,u)$, the roundtrip distance $d(u \leftrightarrow v)$ is simply $d(u,v) + d(v,u)=2d(u,v)$.

This observation naturally leads to the question: what are the potential advantages of studying \emph{roundtrip} routing in undirected graphs? In particular, why might roundtrip routing be easier to approximate than single-message routing, despite the fact that the roundtrip distance is always twice the one-way distance, i.e., $d(u \leftrightarrow v) = 2d(u,v)$?

The key distinction lies in the asymmetry of available information during the routing process. When routing from a source vertex $u$ to a destination vertex $v$, the algorithm has access only to the routing table $\RT(u)$ and the label $\L(v)$. However, routing from $v$ to $u$ is based solely on $\RT(v)$ and $\L(u)$. Since these inputs may differ significantly, the resulting routing paths may differ significantly, and we may have that $\hat{d}(u,v) \neq \hat{d}(v,u)$.

In single-message routing, the stretch must hold for the worst-case direction. Therefore, $\frac{\max\big(\hat{d}(u,v), \hat{d}(v,u)\big)}{d(u,v)}$ is bounded. 
However, roundtrip routing requires only that the average of the two directions be bounded. Therefore,
$\frac{\hat{d}(u\leftrightarrow v)}{d(u \leftrightarrow v)} = \frac{\hat{d}(u,v) + \hat{d}(v,u)}{2d(u,v)}$ is bounded.

This relaxation in the approximation requirement allows us to achieve an optimal stretch compact roundtrip routing scheme, as presented in the following theorem.

\Reminder{T-Compact-Undirected}
First, we describe the preprocessing algorithm, which computes the routing tables and assigns vertex labels. 
The input to the algorithm is a graph $G = (V, E, w)$ and an integer parameter $k > 1$. The algorithm uses~\autoref{L-TZ-Size} to build a hierarchy of vertex sets $A_0, A_1, \ldots, A_k$, where $A_0=V$, $A_k=\emptyset$, $A_{i+1}\subseteq A_i$, $|A_{i}|=n^{1-i/k}$ and $|B_i(u)| = \Ot(n^{1/k})$, for every $0 \leq i \leq k-1$. 
Next, for every  $u\in V$, the preprocessing algorithm computes $B(u)$ and $C(u)$. Then, for every  $u\in V$, the algorithm sets the routing table $\RT(u)$ to:
\[
\RT(u) = \{ \L(u, T(C(v))) \mid v \in B(u) \},
\]  
and the label $\L(u)$ to:
\[
\L(u) = \{ \L(u, T(C(p_i(u)))) \mid i \in [k] \}.
\]
We now turn to bound the size of the routing tables and the vertex labels.
\begin{lemma}\label{L-routing-undirected-RT-and-L-size}
    $|\RT(u)| = \Ot(n^{1/k})$, and $|\L(u)| = \Ot(k)$, for every $u\in V$.
\end{lemma}
\begin{proof}
    From~\autoref{L-TZ-Size} it follows that $|B(u)|=\Ot(n^{1/k})$. From~\autoref{T-Route-In-Trees} it follows that for every tree $T$ it holds that $|\L(u,T)| = \Ot(1)$. Therefore, $|\RT(u)|=|B(u)|\cdot \Ot(1) = \Ot(n^{1/k})$, as required.
    The label of each vertex is composed of $k$ tree labels $L(u,T(C(p_i(u))))$ for every $i\in [k]$. From~\autoref{T-Route-In-Trees} we know that $|L(u,T(C(p_i(u))))|=\Ot(1)$. Therefore, $|\L(u)|=k\cdot \Ot(1)=\Ot(k)$, as required.
\end{proof}

The routing algorithm routes a message from $u$ to $v$ as follows. Let $\ell(x,y)=\min\{i\mid p_{i}(y)\in B(x)\}$ and let $\ell=\ell(u,v)$. 
We route from $u$ to $v$ on the shortest path between $u$ and $v$ in $T(C(p_{\ell}(v)))$. In figure~\ref{fig:undirected-roundtrip} we show the roundtrip route between $u$ and $v$ according to this routing algorithm.

Next, we show that for every $w$ on the shortest path from $u$ to $v$ in $T(C(p_{\ell}(v)))$ we have that $\L(w, T(C(p_{\ell}(v)))) \in \RT(w)$ and therefore $w$ can route to $v$ in $T(C(p_{\ell}(v)))$.
\begin{lemma}\label{L-Routing-undirected-consistent}
    For every $w\in P(u, p_{\ell}(v)) \cup P(p_{\ell}(v), v)$ it holds that $\L(w, T(C(p_{\ell}(v)))) \in \RT(w)$
\end{lemma}
\begin{proof}
If $w \in P(u, p_{\ell}(v))$, then by~\autoref{L-T(C(u))=C(u)-undirected}, we know that $w \in C(p_{\ell}(v))$. Similarly, if $w \in P(p_{\ell}(v), v)$, then from~\autoref{L-T(C(u))=C(u)-undirected}, we know that $w \in C(p_{\ell}(v))$. 
Since  $w \in C(p_{\ell}(v))$ it follows 
from the definition of $C(p_{\ell}(v))$  that $p_{\ell}(v) \in B(w)$. By the definition of $\RT(w)$ since $p_{\ell}(v) \in B(w)$ it follows that $\L(w, T(C(p_{\ell}(v)))) \in \RT(w)$, as required.
\end{proof}

We now turn to the main technical contribution of this section and prove that the stretch of the compact roundtrip routing scheme is $2k-1$.
\begin{lemma}
    \label{L-Routing-undirected-Approximation}
    $\hat{d}(u\leftrightarrow v) \le (2k-1)\cdot d(u \leftrightarrow v)$
\end{lemma}
\begin{proof}
    Let $u, v \in V$, let $b = \ell(u, v) = \min\{i \in [k] \mid p_i(v) \in B(u)\}$, and let $a = \ell(v, u) = \min\{i \in [k] \mid p_i(u) \in B(v)\}$. Assume, without loss of generality, that $b \ge a$. See Figure~\ref{fig:undirected-roundtrip} for an illustration.
    \begin{figure}
    \centering
    \tikzset{every picture/.style={line width=0.75pt}} 
\tikzset{every picture/.style={line width=0.75pt}} 

\begin{tikzpicture}[x=0.75pt,y=0.75pt,yscale=-1,xscale=1]

\draw  [fill={rgb, 255:red, 255; green, 255; blue, 255 }  ,fill opacity=1 ] (57.64,204.79) .. controls (57.64,202.14) and (59.79,200) .. (62.43,200) .. controls (65.07,200) and (67.21,202.14) .. (67.21,204.79) .. controls (67.21,207.43) and (65.07,209.57) .. (62.43,209.57) .. controls (59.79,209.57) and (57.64,207.43) .. (57.64,204.79) -- cycle ;
\draw  [fill={rgb, 255:red, 255; green, 255; blue, 255 }  ,fill opacity=1 ] (181.86,204.79) .. controls (181.86,202.14) and (184,200) .. (186.64,200) .. controls (189.29,200) and (191.43,202.14) .. (191.43,204.79) .. controls (191.43,207.43) and (189.29,209.57) .. (186.64,209.57) .. controls (184,209.57) and (181.86,207.43) .. (181.86,204.79) -- cycle ;
\draw  [fill={rgb, 255:red, 255; green, 255; blue, 255 }  ,fill opacity=1 ] (181.86,83.79) .. controls (181.86,81.14) and (184,79) .. (186.64,79) .. controls (189.29,79) and (191.43,81.14) .. (191.43,83.79) .. controls (191.43,86.43) and (189.29,88.57) .. (186.64,88.57) .. controls (184,88.57) and (181.86,86.43) .. (181.86,83.79) -- cycle ;
\draw  [fill={rgb, 255:red, 255; green, 255; blue, 255 }  ,fill opacity=1 ] (57.86,120.57) .. controls (57.86,117.93) and (60,115.79) .. (62.64,115.79) .. controls (65.29,115.79) and (67.43,117.93) .. (67.43,120.57) .. controls (67.43,123.21) and (65.29,125.36) .. (62.64,125.36) .. controls (60,125.36) and (57.86,123.21) .. (57.86,120.57) -- cycle ;
\draw    (67.21,204.79) -- (179.79,85.96) ;
\draw [shift={(181.86,83.79)}, rotate = 133.45] [fill={rgb, 255:red, 0; green, 0; blue, 0 }  ][line width=0.08]  [draw opacity=0] (8.93,-4.29) -- (0,0) -- (8.93,4.29) -- cycle    ;
\draw    (186.64,88.57) -- (186.64,197) ;
\draw [shift={(186.64,200)}, rotate = 270] [fill={rgb, 255:red, 0; green, 0; blue, 0 }  ][line width=0.08]  [draw opacity=0] (8.93,-4.29) -- (0,0) -- (8.93,4.29) -- cycle    ;
\draw    (181.86,204.79) -- (69.84,122.35) ;
\draw [shift={(67.43,120.57)}, rotate = 36.35] [fill={rgb, 255:red, 0; green, 0; blue, 0 }  ][line width=0.08]  [draw opacity=0] (8.93,-4.29) -- (0,0) -- (8.93,4.29) -- cycle    ;
\draw    (62.64,125.36) -- (62.44,197) ;
\draw [shift={(62.43,200)}, rotate = 270.16] [fill={rgb, 255:red, 0; green, 0; blue, 0 }  ][line width=0.08]  [draw opacity=0] (8.93,-4.29) -- (0,0) -- (8.93,4.29) -- cycle    ;

\draw (46.64,204.57) node [anchor=north west][inner sep=0.75pt]   [align=left] {$\displaystyle u$};
\draw (188.64,207.79) node [anchor=north west][inner sep=0.75pt]   [align=left] {$\displaystyle v$};
\draw (193.64,71) node [anchor=north west][inner sep=0.75pt]   [align=left] {$\displaystyle p_{b}( v)$};
\draw (18.64,109.57) node [anchor=north west][inner sep=0.75pt]   [align=left] {$\displaystyle p_{a}( u)$};

\end{tikzpicture}

    \caption{The roundtrip routing of~\autoref{T-Compact-Undirected}. Let $a=\ell(v,u), b=\ell(u,v)$, where $\ell(x,y)=\min\{i\mid p_{i}(y)\in B(x)\}$. (We assume wlog that $b\ge a$.)}
    \label{fig:undirected-roundtrip}
\end{figure}
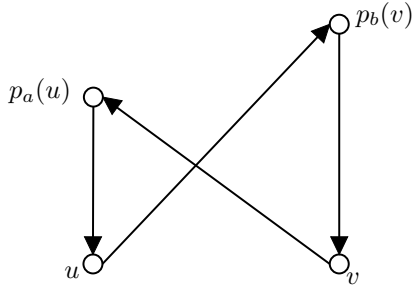

    In the routing phase, we route from $u$  to $v$  on the shortest path between $u$  and $v$  in $T(C(p_{b}(v)))$.
    Similarly,  we route from $v$  to $u$  on the shortest path between $v$  and $u$  in $T(C(p_{a}(u)))$. Therefore,
    \[
        \hat{d}(u,v) = d_{T(C(p_{b}(v)))}(u,v) \le d(u,p_{b}(v)) + d(p_{b}(v),v) \text{ and } \hat{d}(v,u) = d_{T(C(p_{a}(u)))}(u,v) \le d(v,p_{a}(u)) + d(p_{a}(u), u).
    \]
    
    From the triangle inequality, we have 
    $d(u,p_{b}(v)) \le d(u,v) + d(v,p_{b}(v))$, therefore we get that: 
    \[
    \hat{d}(u,v)=d(u,p_{b}(v)) + d(p_{b}(v),v) \le d(u,v) + d(v,p_{b}(v)) + d(p_{b}(v),v) = d(u,v) + d(v \leftrightarrow p_{b}(v))\] 
    By symmetry, we also get that $\hat{d}(v,u) \le d(v,u) + d(u \leftrightarrow p_{a}(u))$.
    By definition $\hat{d}(u \leftrightarrow v) = \hat{d}(u,v) + \hat{d}(v,u)$. Therefore, we get:
    \begin{align*}
    \hat{d}(u \leftrightarrow v) &= \hat{d}(u,v) + \hat{d}(v,u) \\
    &\le d(u,v) + d(v \leftrightarrow p_{b}(v)) + d(v,u) + d(u \leftrightarrow p_{a}(u)) \\
    &= d(u \leftrightarrow v) + d(v \leftrightarrow p_{b}(v)) + d(u \leftrightarrow p_{a}(u)).
    \end{align*}
    
To get that $\hat{d}(u \leftrightarrow v) \le (2k-1) d(u \leftrightarrow v)$, we show that 
$d(v \leftrightarrow p_{b}(v)) + d(u \leftrightarrow p_{a}(u)) \le 2(k-1) d(u \leftrightarrow v)$ in the following claim.
\begin{subclaim}\label{C-lemma-10}
    $d(v \leftrightarrow p_{b}(v)) + d(u \leftrightarrow p_{a}(u)) \le 2(k-1) d(u \leftrightarrow v)$
\end{subclaim}
\begin{proof}
Let $\Delta_i^v = d(v \leftrightarrow p_i(v)) - d(v \leftrightarrow p_{i-1}(v))$, for every $0 < i < k$. 
Notice that for every $0 < j < k$, it holds that:
\[
\sum_{i=1}^j \Delta_i^v = d(v \leftrightarrow p_j(v)) - d(v \leftrightarrow p_{j-1}(v)) + d(v \leftrightarrow p_{j-1}(v)) - \dots + d(v \leftrightarrow p_1(v)) - d(v \leftrightarrow p_0(v)) = d(v \leftrightarrow p_j(v)).
\]
Since we assume (wlog) that $b \ge a$, we have:
\begin{equation}\label{E-1}
d(v \leftrightarrow p_{b}(v)) = \sum_{i=1}^{b} \Delta_i^v = \sum_{i=1}^{a} \Delta_i^v + \sum_{i=a+1}^{b} \Delta_i^v = d(v \leftrightarrow p_{a}(v)) + \sum_{i=a+1}^{b} \Delta_i^v.
\end{equation}

From the definition of $b$ and $a$, for every $0 < i < a \le b$, we have that both $p_i(u) \notin B(v)$ and $p_i(v) \notin B(u)$. Therefore, by applying~\autoref{L-Bound-p-i-d} we get:
\begin{equation}\label{E-2}
d(u \leftrightarrow p_{a}(u)) \le a \cdot d(u \leftrightarrow v) \quad \text{and} \quad d(v \leftrightarrow p_{a}(v)) \le a \cdot d(u \leftrightarrow v).
\end{equation}
For every $a \leq i < b$, since $p_i(v) \notin B(u)$, we can apply~\autoref{L-Bound-delta-2d} to get that $\Delta_i^v < 2d(u \leftrightarrow v)$ . Therefore, we get:
\begin{equation}\label{E-3}
    \sum_{i=a+1}^{b} \Delta_i^v \le \sum_{i=a+1}^{b} 2d(u \leftrightarrow v) = 2(b-a)d(u \leftrightarrow v)
\end{equation}

Using the above three inequalities, we get:
\begin{align*}
    d(u \leftrightarrow p_{a}(u)) + d(v \leftrightarrow p_{b}(v))
    &\stackrel{(\ref{E-1})}{=} d(u \leftrightarrow p_{a}(u)) + d(v \leftrightarrow p_{a}(v)) + \sum_{i=a+1}^{b} \Delta_i^v 
    \\&\stackrel{(\ref{E-2})}{\le} a \cdot d(u \leftrightarrow v) + a \cdot d(u \leftrightarrow v) + \sum_{i=a+1}^{b} \Delta_i^v
    \\&\stackrel{(\ref{E-3})}{\le} 
     a \cdot d(u \leftrightarrow v) + a \cdot d(u \leftrightarrow v)
    + (b - a) \cdot 2d(u \leftrightarrow v) 
    \\&= 2b  d(u \leftrightarrow v) \le 2(k - 1) d(u \leftrightarrow v),
\end{align*}
where the last inequality follows from the fact that $b \le k - 1$.

\end{proof}

Finally, since we know that $\hat{d}(u\leftrightarrow v) \le d(u \leftrightarrow v) + d(v \leftrightarrow p_{b}(v)) + d(u \leftrightarrow p_{a}(u))$, and  from~\autoref{C-lemma-10} we have that $d(v \leftrightarrow p_{b}(v)) + d(u \leftrightarrow p_{a}(u)) \le 2(k-1)d(u \leftrightarrow v)$ we get:
\begin{align*}
   \hat{d}(u\leftrightarrow v) &\le d(u \leftrightarrow v) + d(v \leftrightarrow p_{b}(v)) + d(u \leftrightarrow p_{a}(u))
   \\&\le d(u \leftrightarrow v) + 2(k-1)d(u \leftrightarrow v) = (2k-1)d(u \leftrightarrow v),
\end{align*}
As required.
\end{proof}
\autoref{T-Compact-Undirected} follows from~\autoref{L-routing-undirected-RT-and-L-size},~\autoref{L-Routing-undirected-consistent} and~\autoref{L-Routing-undirected-Approximation}.

\section{Roundtrip routing in Directed graphs} \label{S-roundtrip directed}

\subsection{Extending~\autoref{S-roundtrip-undirected} to directed graphs}
In this section, we consider roundtrip routing in weighted directed graphs. 
One might wonder why the routing scheme of~\autoref{T-Compact-Undirected} does not apply to or cannot be adapted to directed graphs. The key issue lies in the fact that~\autoref{L-T(C(u))=C(u)-undirected} holds only for undirected graphs. Without this lemma, the routing process fails in directed graphs.
Specifically, in directed graphs, there exist cases where $ u \in C(v)$ and $ w \in P(u, v)$, but $ w \notin C(v)$. Consequently, even if we store the set $\{ \L(u, T(C(v))) \mid v \in B(u) \}$,
for every vertex $u$, we may not be able to route correctly from $u$ to $v$ when $ v \in B(u)$.

To overcome this issue, we consider bounded hop diameter graphs, where the hop diameter is the maximum number of edges on a shortest path between any two vertices 
$u,v$ in the graph, i.e. $D_{hop}=\max_{u,v\in V}(|P(u,v)|)$. In such a case, we can achieve the following theorem.
\Reminder{T-Compact-directed-bounded-diameter}
\begin{proof}
The preprocessing algorithm is identical to that of~\autoref{T-Compact-Undirected}, with one key modification: instead of setting 
$
\RT(u) = \{ \L(u, T(C(v))) \mid v \in B(u) \},
$ 
we store 
$
\RT(u) = \{ P(u, v) \mid v \in B(u) \},
$ 
where $ P(u, v) $ is the entire path from $u$ to $v$, and similarly instead of setting
$ \L(u) = \{ p_i(u) \mid i\in [k] \} $ we store $\L(u)=\{P(p_i(u),u) \mid i\in [k]\}$.

In the routing algorithm at the source vertex $u$ to a destination vertex $v$, after determining $\ell = \min\{i\mid p_{i}(v)\in B(u)\}$. The entire path $P(u, p_\ell(v)) \cup P(p_\ell(v),v)$ is attached to the header to ensure that intermediate vertices can route correctly. 
The remainder of the proof follows the same steps as in~\autoref{S-roundtrip-undirected}.
\end{proof}

\subsection{$7$-stretch directed roundtrip routing with $|\RT(u)|=O(n^{1/3})$}\label{S-7-stretch-directed}
In this section, we consider roundtrip routing in general weighted directed graphs. Roditty et. al~\cite{DBLP:journals/talg/RodittyTZ08} obtained a  roundtrip routing scheme with $(4k+\eps)$-stretch and routing tables of size $\Ot(\frac{1}\eps n^{1/k}\log W)$. If we set $k=3$ we get a $(12+\eps)$-stretch  roundtrip routing scheme   with routing tables of size $\Ot(\frac{1}\eps n^{1/3}\log W)$. In this section, we present 
a $7$-stretch roundtrip routing scheme  with routing tables of size $\Ot(n^{1/3})$.
We prove:
\Reminder{T-Compact-Directed-7-1/3}
First, we describe the preprocessing algorithm, which computes the routing tables and assigns vertex labels. The input to the algorithm is a graph $G = (V, E, w)$. 
The algorithm uses~\autoref{L-A-center} and~\autoref{L-TZ-Size} to build a hierarchy of vertex sets $V=A_0 \supseteq A_1 \supseteq A_2 \supseteq A_3=\emptyset$, where $|A_{i}|=n^{1-i/k}$. For every  $u\in V$ it holds that $|C(u, A_1)|=\Ot(n^{1/3})$ and $|B_i(u)| = \Ot(n^{1/k})$, where $0\leq  i \leq 3$.
Next, for every $u\in V$, the preprocessing algorithm computes $B(u)$ and $C(u)$. Then, for every $u\in V$, the algorithm sets the routing table $\RT(u)$ to:

\[
\RT(u) = \left\{ 
\begin{array}{ll}
L(w, T(B_0(u))) & \text{for every } w \in B_0(u) \\
L(u, T(B_0(w))) & \text{for every } w \text{ s.t. } u\in T(B_0(w)) \\
L(u, T(C(w))) & \text{for every } w \in A_2 \\
L(u, T(C(w))) & \text{for every } w \in B_1(u), \text{ if $P(u,w)\subseteq C(w)$}\\
L(w, T(C(p_2(z)))) & \text{for every } w \in  B_1(u), z = \arg\min_{x\in P(u,w), x\notin C(w)}\{d(u, x)\}
\end{array}
\right.
\]

and the label $\L(u)$ to:
\[
\L(u) = \{\L(u, T(B_0(u))), \L(p_1(u), T(B_0(u))), \L(u, T(C(p_2(u))))\}
\]

We now bound the size of the routing tables and the vertex labels.
\begin{lemma}\label{L-RT-and-L-size-directed-7}
    $|\RT(u)| = \Ot(n^{1/3})$ and $|\L(u)|=\Ot(1)$
\end{lemma}
\begin{proof}
    From~\autoref{L-TZ-Size} it follows that $|B_i(u)|=\Ot(n^{1/3})$, for every $1\le i \le 2$. In addition we get that $|A_2|=\Ot(n^{1/3})$. From~\autoref{L-A-center} it follows that $|B_0(u)| = \Ot(n^{1/3})$ and $|C(u,A_1)|=\Ot(n^{1/3})$.
    From~\autoref{T-Route-In-Trees} it follows that for every tree $T$ it holds that $|\L(u,T)| = \Ot(1)$. Therefore, 
    $|\RT(u)| \le \Ot(1) \cdot ( |B_0(u)| + |C(u, A_1)| + |A_2| + |B_1(u)| + |B_1(u)|) = \Ot(n^{1/3})$, as required.
    
    The label of each vertex is composed of three tree labels. From~\autoref{T-Route-In-Trees} it follows that each tree label is of size $\Ot(1)$. Therefore, $|\L(u)| = 3\cdot\Ot(1) = \Ot(1)$, as required.
\end{proof}
The routing algorithm routes a message from $u$ to $v$  as follows. 
If $v\in B_0(u)$, then the algorithm routes the message using the tree $T(B_0(u))$. Otherwise, if $v \in C(u, A_1)$ then the algorithm routes the message using the tree $T(B_0(v))$.
If $v\notin B_0(u)$ and $v\notin C(u, A_1)$, then the algorithm checks if $p_1(v) \in B_1(u)$. 
In this case, if $P(u,p_1(v))\subseteq C(p_1(v))$ then the algorithm routes on $T(C(p_1(v)))$ from $u$ to $p_1(v)$, and then from $p_1(v)$ to $v$ on the tree $T(B_0(v))$ (see Figure~\ref{fig:directed-roundtrip-7} (a)).
Otherwise, if $P(u,p_1(v))\not\subseteq C(p_1(v))$, then the algorithm routes from $u$ to $p_1(v)$ on the tree $T(C(p_2(z)))$, where $z=\arg\min_{x\in P(u,p_1(v)), x\notin C(p_1(v))}\{d(u, x)\}$, and then the algorithm routes from $p_1(v)$ to $v$ on the tree $T(B_0(v))$ (see Figure~\ref{fig:directed-roundtrip-7} (b)).

Finally, if $p_1(v) \notin B_1(u)$, then the algorithm routes from $u$ to $v$ on the tree $T(C(p_2(v)))$ (see Figure~\ref{fig:directed-roundtrip-7} (c)). A pseudo-code for the routing algorithm is given in Algorithm~\ref{Algorithm-Routing-Directed}. 

\begin{algorithm2e}[t] 
\caption{$\Route(u, v)$}\label{Algorithm-Routing-Directed}
\lIf {$v\in B_0(u)$} {
    Route from $u$ to $v$ on $T(B_0(u))$
}
\lIf {$v\in C(u, A_1)$} {
    Route from $u$ to $v$ on $T(B_0(v))$
}
\If{$p_1(v)\in B_1(u)$}{
    \If{$P(u,p_1(v))\subseteq C(p_1(v))$} {
        Route from $u$ to $p_1(v)$ on $T(C(p_1(v)))$ and then route from $p_1(v)$ to $v$ on $T(B_0(v))$.
    }
    \Else {
        $z \gets \arg\min_{x\in P(u,p_1(v)), x\notin C(p_1(v))}\{d(u, x)\}$\;
        Route from $u$ to $p_1(v)$ on $T(C(p_2(z)))$ and then route from $p_1(v)$ to $v$ on $T(B_0(v))$.
    }
}
\Else{
    Route from $u$ to $v$ on $T(C(p_2(v)))$.
}
\end{algorithm2e}  

In the next lemma, we show that all intermediate vertices have the necessary information to complete the routing process once the routing tree is determined.
\begin{lemma}\label{L-routing-is-continuent-directed-7}
The following properties hold:

\begin{enumerate}
    \item If $ v \in B_0(u)$, then for every $ w \in P(u, v) $, we have $ \L(w, T(B_0(u))) \in \RT(w) $. \label{i-b-1}
    \item If $ v \in C(u, A_1) $, then for every $ z \in P(u, v) $, we have $ \L(z, T(B_0(v))) \in \RT(z) $. \label{i-b-2}
    \item If $ p_1(v) \in B_1(u)$ and $ P(u, p_1(v)) \subseteq C(v)$, then for every $ z \in P(u, p_1(v)) $, we have $ \L(z, T(C(p_1(v)))) \in \RT(z) $. \label{i-b-3}
    \item If $ w \in A_2 $, then for every $ z \in P(u, w) \cup P(w, v) $, we have $ \L(z, T(C(w))) \in \RT(z) $. \label{i-b-4}
    \item For every $w \in P(p_1(v), v)$, we have $\L(w, B_0(v)) \in \RT(w) \cup \L(v)$. \label{i-b-5}
\end{enumerate}

\end{lemma}

\begin{proof}
From~\autoref{L-Ball-can-be-tree}, we know that for every $ v \in B(u, X) $ and $ w \in P(u, v) $, we have that $ w \in B(u, X) $. Therefore, properties~\ref{i-b-1} and~\ref{i-b-2} hold.

For property~\ref{i-b-3}, we have that the entire path $ P(u, p_1(v)) $ is contained in $ C(p_1(v))$. Hence, any sub-path of $ P(u, p_1(v)) $ is also in $ C(p_1(v))$, and we have $ \L(z, T(C(p_1(v)))) \in \RT(z) $ for every $ z \in P(u, p_1(v)) $.

For property~\ref{i-b-4}, since $ w \in A_2 $, we know that $ C(w) = V $. Therefore, all vertices in the graph (and specifically $ z$) hold $ \L(z, T(C(w))$.

For property~\ref{i-b-5}, for $ p_1(v)$, we have $ \L(p_1(v), T(B_0(v))) \in \L(v)$. For every other $ w \in P(p_1(v), v) $, we have from~\autoref{L-Digraph-Clusters-are-continuent} that $v\in C(w)$ and therefore $w \in T(B(v))$, and we have $\L(w,T(B_0(v)))$ in $\RT(w)$, as required. 
\end{proof}

We now turn to prove that the stretch of the compact roundtrip routing scheme is $7$.
\begin{lemma}
    \label{L-Routing-directed-Approximation}
    $\hat{d}(u\leftrightarrow v) \le 7\cdot d(u \leftrightarrow v)$
\end{lemma}
\begin{proof}
First, consider the case that $ v \in B_0(u)$. In this case, the routing algorithm routes from $u$ to $v$ along $ P(u, v) $, so $ \hat{d}(u, v) = d(u, v)$. Since $ v \in B_0(u)$, by the definition of $ C(v, A_1) $, we know that $ u \in C(v, A_1) $. Therefore, the routing algorithm from $v$ to $u$ follows $ P(v, u$. Thus, 
\[
\hat{d}(u \leftrightarrow v) = \hat{d}(u, v) + \hat{d}(v, u) = d(u, v) + d(v, u) = d(u \leftrightarrow v),
\]
as required. Similarly, if $ v \in C(u, A_1) $, using symmetrical arguments, we get that $ \hat{d}(u \leftrightarrow v) = d(u \leftrightarrow v) $, as required.

Next, consider the case that $ v \notin B_0(u)$ and $ v \notin C(u, A_1) $. In the following claim, we show that in this case, $ \hat{d}(u, v) \leq d(u, v) + 3d(u \leftrightarrow v) $. Using this claim, we obtain:
\[
\hat{d}(u \leftrightarrow v) = \hat{d}(u, v) + \hat{d}(v, u) \leq d(u, v) + 3d(u \leftrightarrow v) + d(v, u) + 3d(u \leftrightarrow v) \leq 7d(u \leftrightarrow v),
\]
as required.

\begin{subclaim} \label{C-lemma-13}
    If $ v \notin B_0(u)$ and $ v \notin C(u, A_1) $, then:
    \[
    \hat{d}(u, v) \leq d(u, v) + 3d(u \leftrightarrow v).
    \]
\end{subclaim}

\begin{proof}
Since $ v \notin B_0(u) $ and $ v \notin C(u, A_1) $, we know from the definitions of $ B_0(u) $ and $ C(u, A_1) $ that:
\begin{equation}\label{e-bound-p_1}
d(u \leftrightarrow p_1(u)) \leq d(u \leftrightarrow v)
\quad \text{and} \quad
d(v \leftrightarrow p_1(v)) \leq d(u \leftrightarrow v).
\end{equation}

We divide the proof into two cases: the case that $ p_1(v) \notin B(u)$ and the case that $ p_1(v) \in B(u)$.
If $ p_1(v) \notin B(u)$
we can apply~\autoref{L-Bound-delta-2d} and get: 
\begin{equation}\label{e-bound-p2}
    d(v \leftrightarrow p_2(v)) \leq 2d(u \leftrightarrow v) + d(v \leftrightarrow p_1(v)) \stackrel{(\ref{e-bound-p_1})}{\leq} 3d(u \leftrightarrow v).
\end{equation}
Since $ p_1(v) \notin B(u)$, the routing algorithm routes from $u$ to $v$ on the tree $ T(C(p_2(v))$. Therefore, we have:
\[
\hat{d}(u, v) = d_{T(C(p_2(v)))}(u, v) \leq d(u, p_2(v)) + d(p_2(v), v) \stactri\leq d(u, v) + d(v \leftrightarrow p_2(v)) \stackrel{(\ref{e-bound-p2})}\leq d(u, v) + 3d(u \leftrightarrow v),
\]
as required.
\begin{figure}
    \centering
    \tikzset{every picture/.style={line width=0.75pt}} 
\input{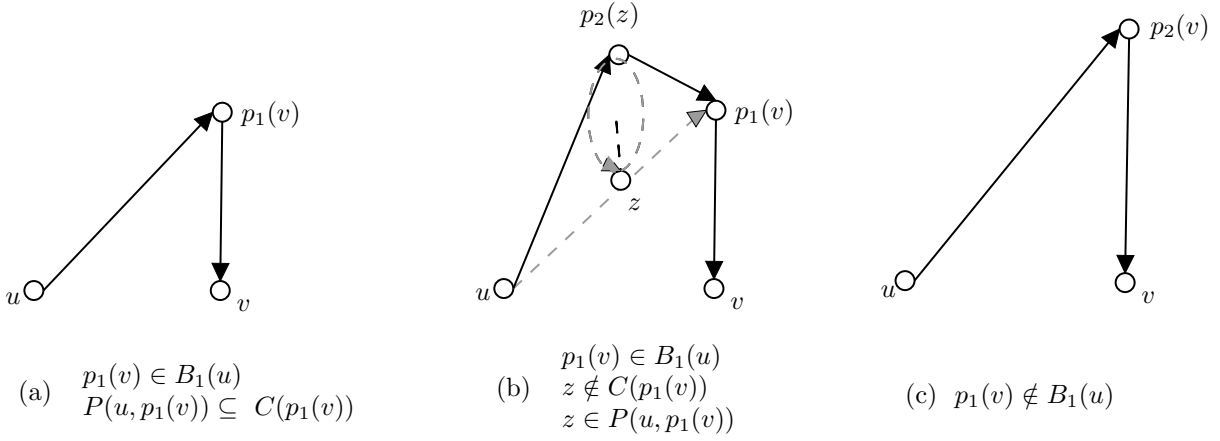}

    \caption{The routing of~\autoref{T-Compact-Directed-7-1/3}. $v\notin B_0(u) \cup C(u, A_1)$.}
    \label{fig:directed-roundtrip-7}
\end{figure}

Next, consider the case that $ p_1(v) \in B_1(u)$.
If $ P(u, p_1(v)) \subseteq C(p_1(v))$, then the routing algorithm routes along the shortest path from $u$ to $ p_1(v)$. Therefore, we have:
\[
\hat{d}(u, v) \leq d(u, p_1(v)) + d(p_1(v), v) \stactri\leq d(u, v) + d(v \leftrightarrow p_1(v)) \stackrel{(\ref{e-bound-p_1})}{\le} d(u,v) + d(u \leftrightarrow v),
\]
as required.

Otherwise, if $ P(u, p_1(v)) \not\subseteq C(p_1(v))$, let $ z $ be the first vertex in $ P(u, p_1(v)) $ such that $ z \notin C(p_1(v))$. 
From the definition of $ C(p_1(v))$, it follows that $ p_1(v) \notin B_1(z$. From the definition of $B_1(z)$ we get that $
d(z \leftrightarrow p_2(z)) \leq d(z \leftrightarrow p_1(v)).
$
Since $ z \in P(u, p_1(v)) $, it follows that
$
d(z \leftrightarrow p_1(v)) \leq d(u \leftrightarrow p_1(v)).
$
Combining these inequalities, we get that:
$
d(z \leftrightarrow p_2(z)) \leq d(u \leftrightarrow p_1(v)).
$

In this case, the routing algorithm routes from $u$ to $ p_1(v)$ on the tree $ T(C(p_2(z))$, and then from $ p_1(v)$ to $v$. We get that:
\begin{align*}
\hat{d}(u, v) &= d_{T(C(p_2(z)))}(u, p_1(v)) + d(p_1(v), v) \\
&\leq d(u, p_2(z)) + d(p_2(z), p_1(v)) + d(p_1(v), v) \\
&\stactri\leq d(u, z) + d(z, p_2(z)) + d(p_2(z), z) + d(z, p_1(v)) + d(p_1(v), v) \\
&= d(u, z) +d(z \leftrightarrow p_2(z)) +  d(z, p_1(v))  + d(p_1(v), v).
\end{align*}
Since $z \in P(u, p_1(v)) $, we know that $d(u, z) + d(z, p_1(v)) = d(u, p_1(v)) $. Recall that $ d(z \leftrightarrow p_2(z)) \leq d(u \leftrightarrow p_1(v)) $. Therefore:
\begin{align*}
\hat{d}(u, v) &\leq d(u, z) + d(z, p_1(v))  + d(p_1(v), v) + d(z \leftrightarrow p_2(z)) \\
&\leq d(u, p_1(v)) + d(p_1(v), v) + d(u \leftrightarrow p_1(v)) \\
&\stactri\leq d(u, v) + d(v \leftrightarrow p_1(v)) + d(u \leftrightarrow v) + d(v \leftrightarrow p_1(v))\\
&\stackrel{(\ref{e-bound-p_1})}{\leq} d(u, v) + d(u \leftrightarrow v) + d(u \leftrightarrow v) + d(u \leftrightarrow v)
= d(u, v) + 3d(u \leftrightarrow v),
\end{align*}
as required.
\end{proof}

From~\autoref{C-lemma-13}, we have that $ \hat{d}(u, v) \leq d(u, v) + 3d(u \leftrightarrow v) $ and $ \hat{d}(v, u) \leq d(u, v) + 3d(u \leftrightarrow v) $. Therefore:
\[
\hat{d}(u \leftrightarrow v) = \hat{d}(u, v) + \hat{d}(v, u) \leq d(u, v) + 3d(u \leftrightarrow v) + d(v, u) + 3d(u \leftrightarrow v) = 7d(u \leftrightarrow v),
\]
as required.
\end{proof}

\autoref{T-Compact-Directed-7-1/3} follows from~\autoref{L-RT-and-L-size-directed-7},~\autoref{L-routing-is-continuent-directed-7} and~\autoref{L-Routing-directed-Approximation}.
\section{$\approx 2.64k$-stretch compact routing scheme for undirected graphs} \label{S-amortized-undirected}
In this section, we consider routing in weighted undirected graphs. We prove:
\Reminder{T-Compact-Amortized-3k}
First, in Section~\ref{S-average-prepr}, we introduce our preprocessing algorithm, along with the routing tables and labels used in our compact routing scheme.
Next, for the sake of simplicity, we describe in Section~\ref{S-average-simple} a routing algorithm that assumes the distance $d(u,v)$ is known at the source vertex $u$ when routing to a destination $v$. In this setting, we present an optimal $(2k-1)$-stretch routing scheme.
In Section~\ref{S-average-full}, we remove the assumption that the distance $d(u,v)$ is known and present our $\approx 2.64k$-stretch routing algorithm.

\subsection{Preprocessing algorithm}\label{S-average-prepr}
The preprocessing algorithm, which computes the routing tables and assigns vertex labels, works as follows. 
The input to the algorithm is a graph $G = (V, E, w)$ and an integer parameter $k > 1$. The algorithm uses~\autoref{L-TZ-Size} to build a hierarchy of vertex sets $A_0, A_1, \ldots, A_k$, where $A_0=V$, $A_k=\emptyset$, $A_{i+1}\subseteq A_i$, $|A_{i}|=n^{1-i/k}$ and $|B_i(u)| = \Ot(n^{1/k})$, for every $0 \leq i \leq k-1$. Let $h_i(u)=d(u,A_i)$.
Next, for every  $u\in V$, the preprocessing algorithm computes $B(u)$ and $C(u)$. Then, for every  $u\in V$, the algorithm sets the routing table $\RT(u)$ to:
\[
\RT(u) = \{ \L(u, T(C(v))) \mid v \in B(u) \} \cup \{L(v,T(C(u))) \mid v\in C(u) \},
\]  
and the label $\L(u)$ to:
\[
\L(u) = \{ \pair{h_i(u), \L(u, T(C(p_i(u))))} \mid i \in [k] \}.
\]

Next, we bound the total size of the routing tables and the size of each label. 
\begin{lemma}\label{L-routing-undirected-average-RT-and-L-size}
    $\sum_{u\in V}|\RT(u)|=\Ot(n^{1+1/k})$, and $|\L(u)| = \Ot(k)$
\end{lemma}
\begin{proof}
    From~\autoref{T-Route-In-Trees}, we know that for every tree $ T $, it holds that $ |\L(u, T)| = \Ot(1)$. Therefore, we get that 
$
\sum_{u \in V} |\RT(u)| = \Ot\left(\sum_{u \in V} \big(|B(u)| + |C(u)|\big)\right)=\sum_{u \in V}(|B(u)|)+\sum_{u \in V}(|C(u)|)
$.
Since $w\in C(u)$ iff $u\in B(u)$, we get that
$
\sum_{u \in V} |C(u)| = \sum_{u \in V} |B(u)|
$. 
Therefore, we get that
$
\sum_{u \in V} |\RT(u)| = 2\sum_{u \in V}(|B(u)|) = \Ot(\sum_{u \in V}(|B(u)|))
$.
From~\autoref{L-TZ-Size}, it follows that $ |B(u)| = \Ot(n^{1/k})$. Therefore, $
\sum_{u \in V} |\RT(u)| = \Ot\left(n \cdot n^{1/k}\right) = \Ot\left(n^{1 + 1/k}\right)
$, as required.
    
    The label of each vertex is composed of $k$ tree labels and $k$ values. From~\autoref{T-Route-In-Trees} we know that $|L(u,T)|=\Ot(1)$. Therefore, $|\L(u)|=k\cdot (\Ot(1)+O(1))=\Ot(k)$, as required.
\end{proof}
\subsection{A $(2k-1)$-stretch routing algorithm assuming  $d(u,v)$ is known}\label{S-average-simple}
In this section, we describe our routing algorithm from the source  vertex~$u$ to the destination vertex~$v$, assuming that the distance~$d(u,v)$ is given to the algorithm together with the routing table~$\RT(u)$, and the label~$\L(v)$. The routing algorithm works as follows. 
For each $0 \le i \le k-1$, the algorithm performs the following steps. If $p_i(v) \in B_i(u)$, it routes from $u$ to $v$ using the tree $T(C(p_i(v)))$. Otherwise, if the inequality $h_{i+1}(v) > d(u,v) + h_i(u)$ holds, the algorithm routes from $u$ to $v$ in the tree $T(C(p_i(u)))$. If neither condition is satisfied, the algorithm proceeds to the next iteration.
A pseudo-code for the routing algorithm is given in Algorithm~\ref{Algorithm-Routing-amortized-simple}.

In the following lemma, we show that if $h_{i+1}(v)>d(u,v) + h_{i}(u)$ then  $v\in C(p_i(u))$ and therefore in such a case it is possible to route 
from $u$ to $v$ in the tree $T(C(p_i(u)))$.

\begin{algorithm2e}[t] 
\caption{$\Route(u, v, d(u,v))$}\label{Algorithm-Routing-amortized-simple}
\For{$i \gets 0$ to $k-1$}{
    \If{$p_i(v)\in B_i(u)$}{Route from $u$ to $v$ in $T(C(p_{i}(v)))$}
    \If{$h_{i+1}(v) > h_{i}(u)+d(u,v)$} {
        Route from $u$ to $v$ in $T(C(p_{i}(u)))$
    }
}
\end{algorithm2e}

\begin{lemma} \label{L-average-simple}
    If $h_{i+1}(v)>d(u,v) + h_{i}(u)$ then $v\in C(p_i(u))$.
\end{lemma}
\begin{proof}
    From the triangle inequality, we have that $d(v,p_i(u)) \le d(v,u)+h_i(u)$. Therefore,  $d(v,p_i(u))< h_{i+1}(v)$. By the definitions of $B(v)$ and $C(p_i(u))$ we have  $p_i(u)\in B(v)$ and $v\in C(p_i(u))$, as required.
\end{proof}

Next, in the following lemma, we show that the stretch of the routing scheme is $(2k-1)$.

\begin{lemma}\label{L-simple-average-corr}
    $\hat{d}(u,v) \le (2k-1)d(u,v)$
\end{lemma}
\begin{proof}
 To prove the lemma, we first prove the following claim.
\begin{subclaim}\label{C-simple-average-corr}
    For every $0\le i\le k$,  either 
$\hat d(u,v)\le(2i-1)d(u,v)$ 
or 
$h_i(u),h_i(v)\le id(u,v)$. 
\end{subclaim}
\begin{proof}
We prove the lemma by induction on $i$.
For the base case $i=0$, we have $h_0(u)=h_0(v)=0$, and therefore the second condition holds.

Next, we assume that the claim holds for $i\ge 0$ and prove it for $i+1$.  
Since the claim holds for $i$, we know that either $\hat d(u,v)\le(2i-1)d(u,v)$ or $h_i(u),h_i(v)\le id(u,v)$.
If $\hat d(u,v)\le(2i-1)d(u,v)$ then $\hat d(u,v)\le(2(i+1)-1)d(u,v)=(2i+1)d(u,v)$ and the claim holds.  Therefore, we consider the case that:
\begin{equation}\label{e--a}
    h_i(u),h_i(v)\le id(u,v)
\end{equation}
In the routing algorithm, if $p_i(v)\in B_i(u)$, then the algorithm routes on the shortest path from $u$ to $v$ in $T(C(p_i(v)))$, and therefore:
$$
\hat d(u,v)\le d(u,p_i(v))+d(p_i(v),v)\stactri\le d(u,v)+2h_i(v)  \stackrel{\ref{e--a}}\le (2i+1)d(u,v).
$$
Otherwise, if $p_i(v)\notin B_i(u)$, then:
$$
h_{i+1}(u)\le d(u,p_i(v)) \stactri\le d(u,v)+h_i(v)\stackrel{\ref{e--a}}\le(i+1)d(u,v).
$$
Therefore we get that either $\hat{d}(u,v) \le (2i+1)d(u,v)$ or $h_{i+1}(u)\le (i+1)d(u,v)$.

In addition, the algorithm checks if $h_{i+1}(v)>d(u,v) + h_{i}(u)$, if this is the case the algorithm routes from $u$ to $v$ in $T(C(p_i(u)))$, and therefore:
$$\hat{d}(u,v) \le d(u,p_i(u)) + d(p_i(u),v) \stactri\le 2h_i(u)+d(u,v) \stackrel{\ref{e--a}}\le (2i+1)d(u,v).$$
Otherwise, we have that $h_{i+1}(v)\le d(u,v) + h_{i}(u) \stackrel{\ref{e--a}}\le (i+1)d(u,v)$, as required.
Overall, we get that either $\hat d(u,v)\le(2(i+1)-1)d(u,v)$ or $h_{i+1}(u),h_{i+1}(v)\le(i+1)d(u,v)$, as required.
\end{proof}

We move to complete the proof of the lemma. Notice that since $A_k=\emptyset$, we have that $h_k(u)=\infty$ and $h_k(v)=\infty$, and therefore the first term of~\autoref{L-simple-average-corr} holds for $i=k$ and therefore $\hat{d}(u,v) \le (2k-1)d(u,v)$, as required.
\end{proof}

\subsection{$\approx 2.64k$-stretch compact routing scheme}\label{S-average-full}
Next, we describe the main modifications to the routing algorithm from Section~\ref{S-average-simple} to transform it into a proper routing algorithm — that is, one that does not rely on the assumption that $d(u,v)$ is known in advance when routing from $u$ to $v$.

Therefore,  we introduce an estimate $\hat{\delta}$ that  always satisfies $\hat{\delta} \leq d(u,v)$.  
Since $\hat{\delta}$ is an estimation of $d(u,v)$ that might be strictly smaller than $d(u,v)$ then  even if 
$h_{i+1}(v)>\hat{\delta} + h_{i}(u)$ it might be that $v\notin C(p_i(u))$. To overcome this hurdle, we use a different condition that when satisfied 
implies that $\hat{\delta} \ll h_{i+1}(v) -h_{i}(u)$, and
we route from $u$ to $p_i(u)$ and check whether $v\in C(p_i(u))$. If  $v\in C(p_i(u))$ then we route from $p_i(u)$ to $v$. Otherwise, if  $v\notin C(p_i(u))$ then we route from $p_i(u)$ back to $u$. 
Since $v \notin C(p_i(u))$, it follows that $h_{i+1}(v) - h_i(u) \leq d(u,v)$.  
Consequently, we can safely update the estimate $\hat{\delta}$ to $h_{i+1}(v) - h_i(u)$,  
while still preserving the invariant $\hat{\delta} \leq d(u,v)$.  
We then continue to the next iteration. 
Notice that since $B(u)$ is known, this problem occurs only on even layers; in odd layers we can simply check whether $p_i(v)\in B(u)$ using the routing table of $u$. Our modification, therefore, revolves around estimating the more complex issue of whether $p_i(u)\in B(v)$.

Next, we formally describe the routing algorithm. The algorithm routes a message from $u$ to $v$ as follows.  
Let $ \ell = \min\{i \mid p_{i}(v) \in B(u)\} $, and let $ a = \lfloor \frac{\ell}{2} \rfloor $, and let $c_0,\ldots,c_a$ be a sequence of constants such that $c_0=1$ and $c_i=2-\frac{a-i}{a+\sum_{j=0}^{i-1}c_j}$, for every $1\le i \le a$. \footnote{The $c_i$ constants follow from the tradeoff presented in \autoref{C-c_i-reason} and they are the reason for the strange stretch of $2.68$.}
Let $\hat{\delta}_0 = \max_{0 \leq i \leq \ell}(h_{i+1}(u) - h_i(v))$. Let $r=0$ and let $j$ be a list, initially empty, that stores indices.

For each $0 \le i \le a$, the algorithm performs the following steps. If $ h_{2i+1}(v) > h_{2i}(u) + c_i \cdot \hat{\delta}_r$, then the algorithm routes from $u$ to $p_{2i}(u)$ on the tree $T(C(p_{2i}(u)))$. 
At the vertex $p_{2i}(u)$, the algorithm accesses $\RT(p_{2i}(u))$ to determine whether $v \in C(p_{2i}(u))$. If this is the case, it retrieves $\L(v, T(C(p_{2i}(u))))$ from $\RT(p_{2i}(u))$ and routes directly from $p_{2i}(u)$ to $v$ along the tree $T(C(p_{2i}(u)))$. Otherwise, if $v \notin C(p_{2i}(u))$, the algorithm routes back from $p_{2i}(u)$ to $u$ using the same tree. In this case, the algorithm increments $r$ by $1$, and sets $j(r) = i$ and $\hat{\delta}_r = h_{2i+1}(v) - h_{2i}(u)$. The algorithm then proceeds to the next iteration.

If the loop ends without finding any $i$ for which $v \in C(p_{2i}(u))$, the algorithm routes from $u$ to $v$ on the tree $T(C(p_{\ell}(v)))$. In our analysis, we let $b$ be the maximum value of $r$ reached in the algorithm.
A pseudo-code for the routing algorithm is given in~\autoref{Algorithm-Routing-amortized}.

\begin{algorithm2e}[t] 
\caption{$\Route(u, v)$}\label{Algorithm-Routing-amortized}
$\ell \gets \min\{i\mid p_{i}(v)\in B(u)\}$ \\
$a \gets \floor{\frac{\ell}2}$ \\
$r \gets 0, j \gets List()$ \\
$\hat{\delta}_0 \gets \max_{0 \leq i < \ell}(h_{i+1}(u)-h_i(v))$ \\
\For{$i \gets 0$ to $a$}{
    \If{$h_{2i+1}(v) > h_{2i}(u)+c_i\hat{\delta}_r$}{ \label{l-cond}
        Route to $p_{2i}(u)$ in $T(C(p_{2i}(u)))$\\
        \If{$v\in C(p_{2i}(u))$}{Route to $v$ in $T(C(p_{2i}(u)))$}\label{l-found}
        \Else{
            $r \gets r+1$ \\ \label{l-update}
            $j(r) \gets i$\\
            $\hat{\delta}_{r} \gets h_{2i+1}(v) - h_{2i}(u)$ \\
            Route to $u$ in $T(C(p_{2i}(u)))$
        }
    }
}
Route to $v$ in $T(C(p_{\ell}(v)))$. \label{l-average-stop}
\end{algorithm2e}

In the following lemma, we show that once the routing tree is determined, every intermediate vertex along the routing path possesses the necessary information to complete the routing process.
\begin{lemma} \label{L-routing-consistent-average-undirected}
The following properties hold:
\begin{enumerate}
    \item For every $w \in P(u, p_i(u))=P(p_i(u),u)$, $\L(w, T(C(p_i(u)))) \in \RT(w)$. \label{P-1}
    \item If $p_i(u) \in B(v)$, then for every $w \in P(p_i(u), v)$, $\L(w, T(C(p_i(u)))) \in \RT(w)$. \label{P-2}
    \item For every $w \in P(u, p_\ell(v))$, $\L(w, T(C(p_\ell(v)))) \in \RT(w)$. \label{P-3}
\end{enumerate}
\end{lemma}

\begin{proof}
We prove each property as follows:

\begin{enumerate}
    \item  
    Let $w \in P(u, p_i(u))$. Since $u \in C(p_i(u))$,~\autoref{L-T(C(u))=C(u)-undirected} implies $w \in C(p_i(u))$. Therefore, $p_i(u) \in B(w)$, and it follows that $\L(w, T(C(p_i(u)))) \in \RT(w)$, as required.

    \item
    Assume $p_i(u) \in B(v)$, and let $w \in P(p_i(u), v)$. Since $p_i(u) \in B(v)$, we have $v \in C(p_i(u))$. By~\autoref{L-T(C(u))=C(u)-undirected}, $w \in C(p_i(u))$, and thus $\L(w, T(C(p_i(u)))) \in \RT(w)$, as required.

    \item
    Let $w \in P(u, p_\ell(v))$. By~\autoref{L-T(C(u))=C(u)-undirected}, $w \in C(p_\ell(v))$. Hence, $\L(w, T(C(p_\ell(v)))) \in \RT(w)$, as required.
\end{enumerate}
\end{proof}

We now turn to analyze the stretch of the compact routing scheme.
\begin{lemma}\label{L-stretch-undirected-average}
    \[
    \hat{d}(u,v) \le 
    \begin{cases}
        (2a+1+2\sum_{i=0}^{a}c_i)d(u,v), & \ell=2a+1 \\
        (2a-1+2\sum_{i=0}^{a}c_i)d(u,v), & \ell=2a
    \end{cases}
    \]
\end{lemma}
\begin{proof}
To prove the lemma, we first prove the following six claims. 
First, we prove that the longest routing path occurs when the algorithm stops at~\autoref{l-average-stop}. Therefore, throughout the proof, we assume that the algorithm does not reach~\autoref{l-found} and does not route from $p_{2i}(u)$ to $v$ in the tree $T(C(p_{2i}(u)))$.

\begin{subclaim}
    The longest routing path occurs when the algorithm stops at~\autoref{l-average-stop}.
\end{subclaim}
\begin{proof}
    If the algorithm, at a vertex $p_{2i}(u)$, determines that $v \in C(p_{2i}(u))$, then the additional routing distance is $d(p_{2i}(u), v) \stackrel{\triangle}{=} h_{2i}(u) + d(u, v)$. However, if the algorithm finds that $v \notin C(p_{2i}(u))$, then it routes back to $u$, adding $h_{2i}(u)$ to the routing distance, and still needs to route from $u$ to $v$, which requires at least $d(u, v)$.
\end{proof}

In the following claim, we show that $\hat{\delta}_r \le d(u,v)$ throughout the routing algorithm.
\begin{subclaim}\label{C-delta-smaller}
    $\hat{\delta}_r \le d(u,v)$, for every $0 \le r \le b$.
\end{subclaim}
\begin{proof}
    For every $1 \le r \le b$, from~\autoref{C-lower-bound-on-delta_r} we have that $\hat{\delta}_{r} \ge \hat{\delta}_{r-1}$ since $\hat{\delta}_{r} \ge c_{j(r)}\hat{\delta}_{r-1}$ and $c_{j(r)} \ge 1$. Therefore, it suffices to show that $\hat{\delta}_b \le d(u,v)$.
    
    We divide the proof into two cases.
    If $b=0$, then we know that $\hat{\delta}_b=\hat{\delta}_0=\max_{0\le i< \ell}(h_{i+1}(u) - h_{i}(v))$. Let $x$ be the index such that $\hat{\delta}_0=h_{x+1}(u)-h_{x}(v)$.
    From the definition of $\ell$ and since $x<\ell$, we know that $p_{x}(v) \notin B(u)$. Therefore $h_{x+1}(u) \le d(u, p_x(v))$. From the triangle inequality, we know that $d(u, p_x(v)) \le d(u,v) + h_x(v)$. By subtracting $h_x(v)$, we get that $h_{x+1}(u)-h_x(v) \le d(u, p_x(v))-h_x(v) \le d(u,v)$. Since $\hat{\delta}_0 = h_{x+1}(u)-h_x(v)$ we get that $\hat{\delta}_0 \le d(u,v)$, as required.
    
    Next, we consider the case where $ b > 0 $. In this case, during the $ j(b) $-th iteration, it holds that $ v \notin C(p_{2j(b)}(u)) $. Therefore, $ p_{2j(b)}(u) \notin B(v) $, and we have that $ h_{2j(b)+1}(v) \leq d(v, p_{2j(b)}(u)) $. Using the triangle inequality, we get that 
    $d(v, p_{2j(b)}(u)) \leq h_{2j(b)}(u) + d(u, v)$.
    Therefore, 
    $h_{2j(b)+1}(v) - h_{2j(b)}(u) \leq d(u, v)$, since $ \hat{\delta}_b = h_{2j(b)+1}(v) - h_{2j(b)}(u) $, it follows that 
    $\hat{\delta}_b \leq d(u, v)$,
    as required.
\end{proof}
In the following claim, we show that $\hat{\delta}_r \ge c_{j(r)}\cdot \hat{\delta}_{r-1}$.
\begin{subclaim}\label{C-lower-bound-on-delta_r}
    $\hat{\delta}_r \ge c_{j(r)}\cdot \hat{\delta}_{r-1}$
\end{subclaim}
\begin{proof}
    Since we updated the value of $\hat{\delta}_r$, we know that the condition in line~\ref{l-cond} is false, and therefore $h_{2i+1}(v) \ge h_{2i}(u) + c_{j(r)}\hat{\delta}_{r-1}$. Thus, $h_{2i+1}(v) - h_{2i}(u) \ge c_{j(r)}\hat{\delta}_{r-1}$.
    Since $\hat{\delta}_{r} = h_{2i+1}(v) - h_{2i}(u)$ we get that 
    and that $\hat{\delta}_{r} = h_{2i+1}(v) - h_{2i}(u) \ge c_{j(r)}\hat{\delta}_{r-1}$, as required. 
\end{proof}
In the following claim, we bound the value of $h_{2i+2}(u)$ in the case that the condition of line~\ref{l-cond} is false in the $i$-th iteration of the loop.
\begin{subclaim}\label{C-cond-never-holds}
    Let $0\le i < a$. Let $r_i$ be the value of $r$ during the $i$-th iteration.
    If $h_{2i+1}(v) \leq h_{2i}(u)+c_i\hat{\delta}_{r_i}$ then
    \[
    h_{2i+2}(u) \le h_{2i}(u) + \hat{\delta}_0+c_i\hat{\delta}_{r_i}
    \]
\end{subclaim}
\begin{proof}
    From the definition of $\hat{\delta}_0$ we know that $\hat{\delta}_0 \ge h_{2i+2}(u)-h_{2i+1}(v)$, therefore: 
    \begin{equation}\label{e-2}
        h_{2i+2}(u) \le h_{2i+1}(v)+\hat{\delta}_0
    \end{equation}
    From the assumption of the lemma we know that $h_{2i+1}(v) < h_{2i}(u)+c_i\hat{\delta}_{r_i}$.
    Overall we get that:
    \begin{align*}
         h_{2i+2}(u) &\stackrel{(\ref{e-2})}{\le} h_{2i+1}(v)+\hat{\delta}_0
        \le h_{2i}(u)+c_{i}\hat{\delta}_{r_i} +\hat{\delta}_0,
    \end{align*}
    as required.
\end{proof}

In the following claim, we present an equality that follows from the definition of the sequence $c_i$.
\begin{subclaim}\label{C-c_i-reason}
For every $1\leq j \leq a$ it holds that $\frac{2j+2a+4\sum_{i=0}^{j-1}c_{i}}{c_j} = 2a+2\sum_{i=0}^{j-1}c_i$.
\end{subclaim}
\begin{proof}
From the definition of $c_j$ we know that 
$c_j = 2-\frac{a-j}{a+\sum_{i=0}^{j-1}c_i}$, thus:
    \begin{align*}
        c_j &= 2-\frac{a-i}{a+\sum_{i=0}^{j-1}c_i} = \frac{2(a+\sum_{i=0}^{j-1}c_i) - (a-i)}{a+\sum_{i=0}^{j-1}c_i} \\
        c_j &= \frac{a+i+2\sum_{i=0}^{j-1}c_i}{a+\sum_{i=0}^{j-1}c_i} = \frac{2a+2i+4\sum_{i=0}^{j-1}c_i}{2a+2\sum_{i=0}^{j-1}c_i} \\
        2a+2\sum_{i=0}^{j-1}c_i &= \frac{2a+2i+4\sum_{i=0}^{j-1}c_i}{c_j},
    \end{align*}
    as required.
\end{proof}

Next, we present the main technical contribution of this section in the following claim.
\begin{subclaim}\label{C-final-proof}
    \[
    2h_{2a}(u) + \sum_{r=1}^{b} 2h_{2j(r)}(u) \le (2a+2\sum_{i=0}^{a-1}c_i)\hat{\delta}_{b}
    \]
\end{subclaim}
\begin{proof}
    We define an instance of the algorithm by fixing the values of $h_{i}(u)$ and $h_{i}(v)$, for every $1\le i \le \ell$.
    Notice that since we assume that the algorithm stops at~\autoref{l-average-stop}, the list $j$ and the value of $b$ depend only on the values of $h_1(u),\dots,h_\ell(u)$ and $h_1(v),\dots,h_\ell(v)$.
        
    We prove the claim by induction on $b$.
    For the base of the induction ($b = 0$), we know that the condition on line~\ref{l-cond} never holds, which means that for every $1\le i \le a$, it holds that $h_{2i+1}(v) \leq h_{2i}(u)+c_i\hat{\delta}_r$.
    Therefore, we can apply~\autoref{C-cond-never-holds} 
    for every $i$ from $a-1$ to $0$ to bound $h_{2i+2}(u)$ from above with $h_{2i}(u) + \hat{\delta}_0+c_i\hat{\delta}_{0}$ and get that:
    \[
    2h_{2a}(u) \le 2\left( \hat{\delta}_0 + c_{a-1}\hat{\delta}_0 + h_{2a-2}(u)\right) \le \dots \le 2\left(\sum_{i=0}^{a-1}(1+ c_{i})\hat{\delta}_0\right) = \left(2a+2\sum_{i=0}^{a-1} c_{i}\right) \hat{\delta}_0,
    \]
    as required.

    Next, we move to prove the induction step. We assume the claim holds for every instance whose list $j$ is of length $b-1$ and prove the claim for an instance whose list $j$ is of length $b$.
    Let  $I$ be an instance whose values are $h_1(u),\dots,h_\ell(u)$ and $h_1(v),\dots,h_\ell(v)$ and its list $j$ is of length $b$.
    Let  $I'$ be an instance whose values $h'_1(u),\dots,h'_\ell(u)$ and $h'_1(v),\dots,h'_\ell(v)$ satisfy the following:
    \begin{itemize}
        \item $h_{i}'(u)=h_{i}(u)$ and $h_{i-1}'(v)=h_{i-1}(v)$, for every $1 \le i \le 2j(b)$.
        \item $h_{i+1}'(u)=h_{i}'(v)+\hat{\delta}_{0}$, for every $2j(b) \leq  i < \ell$.
        \item $h_{2i+1}'(v)=h_{2i}'(u)+c_{i}\hat{\delta}_{b-1}$ and $h_{2i}'(v)=h_{2i-1}'(u)+c_{i}\hat{\delta}_{b-1}$, for every $j(b) \le i\le a-1$.
    \end{itemize}
    To distinguish between the variables of $I$ and $I'$, we denote with $b'$, $j'$, $\hat{\delta}_r'$ the values of $b$, $j$, $\hat{\delta}_r$ in the instance $I'$.
    In the following claim, we prove the connections between the instance $I$ and the instance $I'$. We show:
    \begin{subsubclaim}\label{C-'-properties}
    Instance $I'$ satisfies the following properties:
    \begin{enumerate}[(i)]
        \item $\hat{\delta}_0'=\hat{\delta}_0$.
        \item $j'(r)=j(r) \text{ and } \hat{\delta}'_{r}=\hat{\delta}_{r}, \text{ for every } 0 \le r \le b-1$. \label{p-'=regular}
        \item $b'=b-1$. \label{p-b'=b-1}
        \item $h_{2a}'(u) = h_{2j(b)}(u)+\sum_{i=j(b)}^{a-1}(\hat{\delta}_0 + c_i\hat{\delta}_{b-1})$. \label{p-h_2a'}
    \end{enumerate}
    \end{subsubclaim}
    \begin{proof}
    \begin{enumerate}[(i)]
        \item 
        From the definition of $\hat{\delta}'_0$ we know that $\hat{\delta}'_0 = \max_{0< x < \ell}(h_x'(u)-h_{x-1}'(v))$.
        We divide the possible values of $x$ into two ranges: $1\le x \le 2j(b)$ and $2j(b)<x<\ell$.
        For every $1\le x\le 2j(b)$, since $h_x'(u)=h_x(u)$ and $h_{x-1}'(v)=h_{x-1}(v)$ we have that $\max_{0< x < 2j(b)}(h_x'(u)-h_{x-1}'(v))=\max_{0< x < 2j(b)}(h_x(u)-h_{x-1}(v))$.
        For every $2j(b) < x < \ell$  we have that $h_{x}'(u)-h_{x-1}'(v)=\hat{\delta}_0$. Therefore, we get that $\max_{2j(b)< x < \ell}(h_x'(u)-h_{x-1}'(v)) = \hat{\delta}_0$. Overall we get:
        \begin{align*}
            \hat{\delta}'_0 = \max_{0< x < \ell}(h_x'(u)-h_{x-1}'(v)) &= \max\left(\max_{0< x \le 2j(b)}(h_x'(u)-h_{x-1}'(v)), \max_{2j(b) < x < \ell}(h_x'(u)-h_{x-1}'(v))\right) 
            \\&= \max\left(\max_{0< x \le 2j(b)}(h_x(u)-h_{x-1}(v)), \hat{\delta}_0 \right) = \hat{\delta}_0,
        \end{align*}
        where the last equality follows from the fact that $h_x(u)-h_{x-1}(v) \le \hat{\delta}_0$, for every $1 < x \le 2j(b)$. 
    \item Since $\hat{\delta}'_0=\hat{\delta}_0$, and for every $1 \le i \le j(b)$ we have that $h_{2i}'(u)=h_{2i}(u)$ and $h_{2i-1}'(v)=h_{2i-1}(v)$ we get that the first $j(b)-1$ iterations of the for loop of instance $I$ are the same as those of instance $I'$. Therefore, we get that:
        $j'(r)=j(r) \text{ and } \hat{\delta}'_{r}=\hat{\delta}'_{r}, \text{ for every } 0 \le r \le b-1$.
    \item Since $h_{2i+1}'(v)=h_{2i}'(u)+c_i\hat{\delta}_{b-1}$, for every $j(b) \le i\le a-1$, we get that the condition on line~\ref{l-cond} does not hold for every $j(b) \le i\le a-1$, and therefore $b'=b-1$.
    \item Since for every $j(b) \le i\le a-1$ we have that $h_{2i+2}'(u)=h_{2i}'(u)+\hat{\delta}_{0}+c_i\hat{\delta}_{b-1}$, we get:
    \[
    h_{2a}'(u) = h_{2a-2}'(u)+\hat{\delta}_0+c_{a-1}\hat{\delta}_{b-1}=\dots=h_{2j(b)}'(u)+\sum_{i=j(b)}^{a-1}(\hat{\delta}_0 + c_i\hat{\delta}_{b-1})=h_{2j(b)}(u)+\sum_{i=j(b)}^{a-1}(\hat{\delta}_0 + c_i\hat{\delta}_{b-1}).
    \]
    \end{enumerate}
    \end{proof}
    
    From~\autoref{C-'-properties} (\ref{p-b'=b-1}) we know that $b'=b-1$, and therefore we can apply the induction assumption on the instance $I'$ instance and get that:
    \begin{align*}
        2h_{2a}'(u) + \sum_{{r}=1}^{b-1} 2h_{2j'({r})}'(u) \le (2a + 2\sum_{i=0}^{a-1} c_i) \hat{\delta}'_{b-1}
    \end{align*}
    From~\autoref{C-'-properties} (\ref{p-'=regular}), for every $1\le r \le b-1$, we know that $j'(r)=j(r)$,  and since $j(r) < j(b)$ we know that $2h_{2j'({r})}'(u)=2h_{2j({r})}(u)$. 
    From~\autoref{C-'-properties} (\ref{p-'=regular}) we have that $\hat{\delta}'_{b-1}=\hat{\delta}_{b-1}$. Therefore:
    \begin{equation}\label{e-81}
        2h_{2a}'(u) + \sum_{{r}=1}^{b-1} 2h_{2j({r})}(u) \le (2a + 2\sum_{i=0}^{a-1} c_i) \hat{\delta}_{b-1}
    \end{equation}
    From~\autoref{C-'-properties} (\ref{p-h_2a'}) we know that $h_{2a}'(u)=h_{2j(b)}(u)+\sum_{i=j(b)}^{a-1}(\hat{\delta}_0 + c_i\hat{\delta}_{b-1})$, and therefore we get:
    \[
    2\left(h_{2j(b)}(u)+\sum_{i=j(b)}^{a-1}(\hat{\delta}_0 + c_i\hat{\delta}_{b-1})\right) + \sum_{{r}=1}^{b-1} 2h_{2j({r})}(u) \le (2a + 2\sum_{i=0}^{a-1} c_i) \hat{\delta}_{b-1} 
    \]
    \begin{equation}\label{e-lala0}
        2\sum_{i=j(b)}^{a-1}(\hat{\delta}_0 + c_i\hat{\delta}_{b-1}) + \sum_{{r}=1}^{b} 2h_{2j({r})}(u) \le (2a + 2\sum_{i=0}^{a-1} c_i)\hat{\delta}_{b-1}
    \end{equation}
    Next, we turn to bound $h_{2a}(u)$, by applying~\autoref{C-cond-never-holds} $(a-j(b))$ times. We get:
    \begin{equation}\label{e-lala}
    h_{2a}(u)\le \hat{\delta}_0+c_{a-1}h_{2a-2}(u)\hat{\delta}_b \le \dots \le h_{2j(b)}(u)+\sum_{i=j(b)}^{a-1}(\hat{\delta}_0+c_i\hat{\delta}_b).
    \end{equation}

    Next, we bound $2h_{2a}(u)  + \sum_{{r}=1}^{b} 2h_{2j({r})}(u)$ and complete the proof of the claim:
    \begin{align*}
        2h_{2a}(u)  + \sum_{{r}=1}^{b} 2h_{2j({r})}(u) &\stackrel{(\ref{e-lala})}{\le}
        2\left(h_{2j(b)}(u)+\sum_{i=j(b)}^{a-1}(\hat{\delta}_0+c_i\hat{\delta}_b)\right) + \sum_{{r}=1}^{b} 2h_{2j({r})}(u)
        \\&=
        2h_{2j(b)}(u)+2\sum_{i=j(b)}^{a-1}(\hat{\delta}_0+c_i\hat{\delta}_b-c_i\hat{\delta}_{b-1}+c_i\hat{\delta}_{b-1}) + \sum_{{r}=1}^{b} 2h_{2j({r})}(u)
        \\&= 2h_{2j(b)}(u) + 2\sum_{i=j(b)}^{a-1}c_i\hat{\delta}_b-2\sum_{i=j(b)}^{a-1}c_i\hat{\delta}_{b-1} + 2\sum_{i=j(b)}^{a-1}(\hat{\delta}_0+c_i\hat{\delta}_{b-1})+ \sum_{{r}=1}^{b} 2h_{2j({r})}(u)
        \\&\stackrel{(\ref{e-lala0})}{\le} 2h_{2j(b)}(u) + 2\sum_{i=j(b)}^{a-1}c_i\hat{\delta}_b-2\sum_{i=j(b)}^{a-1}c_i\hat{\delta}_{b-1} + (2a + 2\sum_{i=0}^{a-1} c_i)\hat{\delta}_{b-1}
        \\&= 2h_{2j(b)}(u) + 2\sum_{i=j(b)}^{a-1}c_i\hat{\delta}_b + (2a + 2\sum_{i=0}^{j(b)-1} c_i)\hat{\delta}_{b-1}
    \end{align*}
    By applying~\autoref{C-cond-never-holds} $j(r)$ times, we get that $h_{2j(b)}(u) \le j(b)\hat{\delta}_0 + \sum_{{r}=1}^{b}\sum_{i=j({r}-1)}^{j({r})-1}c_i\hat{\delta}_{{r}-1}$. Since $\hat{\delta}_{{r}-1} \le \hat{\delta}_{b-1}$, for every $1\le r\le b$, we have that $h_{2j(b)}(u) \le (j(b)+\sum_{i=0}^{j(b)-1}c_i)\hat{\delta}_{b-1}$. By applying it we get:
    \begin{align*}
        2h_{2a}(u)  + \sum_{{r}=1}^{b} 2h_{2j({r})}(u) &\le   2(j(b)+\sum_{i=0}^{j(b)-1}c_i)\hat{\delta}_{b-1} + 2\sum_{i=j(b)}^{a-1} c_i\hat{\delta}_{b} + (2a+2\sum_{i=0}^{j(b)-1}c_i)\hat{\delta}_{b-1}
        \\&= 2\sum_{i=j(b)}^{a-1} c_i\hat{\delta}_{b} + (2a+2j(b)+4\sum_{i=0}^{j(b)-1}c_i)\hat{\delta}_{b-1}
        \\&\stackrel{C\ref{C-lower-bound-on-delta_r}}{\le} 2\sum_{i=j(b)}^{a-1} c_i\hat{\delta}_{b} + \frac{(2a+2j(b)+4\sum_{i=0}^{j(b)-1}c_i)\hat{\delta}_{b}}{c_{j(b)}}
        \\&\stackrel{C\ref{C-c_i-reason}}{\le} 2\sum_{i=j(b)}^{a-1} c_i\hat{\delta}_{b} + 2a+2\sum_{i=0}^{j(b)-1}c_i\hat{\delta}_b = (2a+2\sum_{i=0}^{a-1}c_i)\hat{\delta}_b,
    \end{align*}
    as required.
\end{proof}

After proving all the required claims, we can finally turn to prove the lemma.
In algorithm $\Route$, there are two possible routing scenarios. The first scenario is that there exists $r\in [1,b]$ such that $v\in C(p_{2j(r)}(u))$. In such a case and the algorithm routes for every $1\le r' < r$  from $u$ to $p_{2j(r')}(u)$ and back and then the algorithm routes from $u$ to $p_{2j(r)}(u)$ and from $p_{2j(r)}(u)$ to $v$. The second scenario is that for every $r\in [1,b]$ it holds that $v\notin C(p_{2j(r)}(u))$.
In such a case and the algorithm routes for every $1\le  r \leq b$  from $u$ to $p_{2j(r)}(u)$ and back and then the algorithm routes from $u$ to $p_{\ell}(v)$ and from $p_{\ell}(v)$ to $v$. Notice that in the second scenario, the message traverses at least as much as it traverses in the first scenario and therefore it suffices to bound only the second scenario.
In the second scenario, we have :
\begin{equation}\label{e-final-routing}
    \hat{d}(u,v)  \le \sum_{{r}=1}^{b} 2h_{2j(r)}(u) + d(u,p_\ell(v)) + d(p_\ell(v),v) \stactri\le \sum_{{r}=1}^{b} 2h_{2j(r)}(u) + 2h_\ell(v) + d(u,v).
\end{equation}
We divide the proof into two cases, the case that $\ell=2a+1$ and the case that $\ell=2a$. Assume $\ell=2a+1$. 
Since we are in the second scenario we have $h_{2a+1}(v)\le h_{2a}(u) + c_a \hat{\delta}_b$. 
We get:
\begin{align*}
\hat{d}(u,v) &\stackrel{(\ref{e-final-routing})}{\le} \sum_{{r}=1}^{b} 2h_{2j(r)}(u) + 2h_\ell(v) + d(u,v) \le \sum_{{r}=1}^{b} 2h_{2j(r)}(u) + 2h_{2a}(u) + 2c_a \hat{\delta}_b + d(u,v) 
\\&\stackrel{C\ref{C-final-proof}}{\le} (2a+2\sum_{i=0}^{a-1}c_i)\hat{\delta}_{b} + 2c_a \hat{\delta}_b + d(u,v) \stackrel{C\ref{C-delta-smaller}}{\le} (2a+2\sum_{i=0}^{a}c_i + 1)d(u,v)
\end{align*}

Assume that $\ell=2a$. From the definition of $h_{2a}(v)$, we know that $h_{2a}(v) = d(v,p_{2a}(v)) \le d(v,p_{2a}(u))$. From the triangle inequality, we have that $h_{2a}(v) \le h_{2a}(u)+d(u,v)$. We get:
\begin{align*}
    \hat{d}(u,v) &\stackrel{(\ref{e-final-routing})}{\le}  \sum_{{r}=1}^{b} 2h_{2j(r)}(u) + 2h_\ell(v) + d(u,v) \stactri{\le} \sum_{{r}=1}^{b} 2h_{2j(r)}(u) + 2h_{2a}(u)+2d(u,v) + d(u,v) 
    \\& \stackrel{C\ref{C-final-proof}}{\le} (2a+2\sum_{i=0}^{a-1}c_i)\hat{\delta}_{b} + 3d(u,v) \stackrel{C\ref{C-delta-smaller}}{\le} (2a+2\sum_{i=0}^{a-1}c_i + 3)d(u,v) \stackrel{c_a=2}{=} (2a+2\sum_{i=0}^{a-1}c_i - 1)d(u,v),
\end{align*}
as required.
\end{proof}
Since $B_{k-1}(u)=A_{k-1}$, we get that $p_{k-1}(v)\in B(u)$ and therefore $\ell \le k-1$. Therefore, in the worst case of our routing, we have that $a=\floor{\frac{k-1}{2}}$ and get the following stretch.
\begin{corollary}\label{cor-average-stretch}
    Let $a=\floor{\frac{k-1}{2}}$.
    \[
    \hat{d}(u,v) \le 
    \begin{cases}
        (2\sum_{i=0}^{a}c_i+2a+1)d(u,v), & k=2a+2, \\
        (2\sum_{i=0}^{a}c_i+2a-1)d(u,v), & k=2a+1.
    \end{cases} 
    \]
\end{corollary}

\autoref{T-Compact-Amortized-3k} follows from~\autoref{L-routing-undirected-average-RT-and-L-size},~\autoref{L-routing-consistent-average-undirected} and~\autoref{cor-average-stretch}.
Next, we analyze the stretches of~\autoref{T-Compact-Amortized-3k} for different values of $ k $. To compute the stretch for each value of $ k $, we use the following Python code to follow the series $c_i=2-\frac{a-i}{a+\sum_{j=0}^{i-1}c_j}$ and obtain the stretch of the compact routing scheme for $ k $.
\begin{lstlisting}
def f(k):
    a = int((k - 1) / 2)  # \ell <= k-1
    c_j = [1] # c_0 = 1
    sum_of_c_i = [1] 
    for i in range(1,a+1,1):
        c_j.append(2-(a-i)/(a+sum_of_c_i[-1]))
        sum_of_c_i.append(sum_of_c_i[-1] + c_j[-1])
    return 2 * sum_of_c_i[-1] + 2 * a + (-1)**(k % 2)
\end{lstlisting}
In~\autoref{tab:compact_amortized_average_stretches} we present the stretches obtained for various values of $k$.
\begin{table}[t]
    \centering
    \begin{tabular}{|c|c|c|c|c|c|c|c|c|c|c|}
        \hline
        $k$ & $4$  & $6$ & $8$  & $10$ & $20$  & $100$ & $10,000$ & $10^6$ & $10^9$\\
        \hline
        $\frac{\hat{d}(u,v)}{d(u,v)}$ & $ 9.0 $ & $ 14.3 $ & $ 19.6 $ & $ 24.9 $ & $ 51.3 $ & $ 262.4 $ & $ 26372.8 $ & $ 2637412.5 $ & $ 2637413837.2 $\\
        \hline
        $\frac{\hat{d}(u,v)}{k\cdot d(u,v)}$ & $ 2.250 $ & $ 2.389 $ & $ 2.455 $ & $ 2.493 $ & $ 2.567 $ & $ 2.624 $ & $ 2.637 $ & $ 2.637 $ & $ 2.637 $\\
        \hline
    \end{tabular}
    \caption{Stretches of~\autoref{T-Compact-Amortized-3k}}
    \label{tab:compact_amortized_average_stretches}
\end{table}


\bibliographystyle{plainurl}
\bibliography{bibliography}
\end{document}